\documentclass{article}
\usepackage{ijcai23}

\usepackage{times}
\usepackage{soul}
\usepackage{url}
\usepackage[hidelinks]{hyperref}
\usepackage[utf8]{inputenc}
\usepackage[small]{caption}
\usepackage{graphicx}
\usepackage{amsmath}
\usepackage{amsthm}
\usepackage{booktabs}
\usepackage{tabularx}
\usepackage{multirow}
\usepackage{color}
\usepackage{algorithm}
\usepackage[noend]{algpseudocode}
\usepackage{amssymb}
\usepackage{amsbsy}
\usepackage{bbding}
\usepackage{multicol}
\usepackage{multirow}
\usepackage{scalerel}
\usepackage{subcaption}
\usepackage{color}
\usepackage{xcolor}
\usepackage[toc,page]{appendix}
\usepackage{colortbl}
\usepackage{diagbox}
\usepackage{stfloats}
\usepackage{xpatch}
\usepackage{tablefootnote}
\makeatletter
\newif\ifdiagbox@cellEmpty@

\define@key{diagbox}{widest}{%
  \ifdiagbox@cellEmpty@
    \setkeys{diagbox}{innerwidth=\widthof{#1}}%
    \PackageWarning{diagbox}{innerwidth is set to \the\diagbox@wd}%
  \fi}

\define@key{diagbox}{highest}{%
  \ifdiagbox@cellEmpty@
    \setkeys{diagbox}{height=#1}%
  \fi}

\define@key{diagbox}{text}{%
  \def\diagbox@text{#1}}

\define@key{diagbox}{align}{%
  \in@{#1}{l, c, r}%
  \ifin@
    \def\diagbox@align{#1}%
  \else
    \PackageError{diagbox}%
      {The value of key "text" must be one of "l", "c", and "r".}%
  \fi}

\xpatchcmd{\diagbox@double}{%
  \setkeys{diagbox}{dir=NW,#1}%
}{%
  \if\relax\detokenize{#2}\relax
    \if\relax\detokenize{#3}\relax
      \diagbox@cellEmpty@true
      \setkeys{diagbox}{highest=1\line, align=l, text=\@empty}%
    \fi
  \fi
  \setkeys{diagbox}{dir=NW, #1}%
  \ifdiagbox@cellEmpty@
    \rlap{\makebox
      [\dimexpr\diagbox@wd-\diagbox@insepl-\diagbox@insepr\relax]%
      [\diagbox@align]%
      {\diagbox@text}}%
  \fi
}{}{\ddt}
\makeatother

\newfloat{figtab}{htb}{fgtb}
\makeatletter
  \newcommand\figcaption{\def\@captype{figure}\caption}
  \newcommand\tabcaption{\def\@captype{table}\caption}
\makeatother

\newtheorem{assumption}{Assumption}
\newtheorem{theorem}{Theorem}
\newtheorem{lem}{Lemma}
\newtheorem{definition}{Definition}
\newtheorem{prop}{Proposition}
\newtheorem{rmk}{Remark}

\newcommand{\calA}{\mathcal{A}}
\newcommand{\calD}{\mathcal{D}}

\newcommand{\calU}{\mathcal{U}}

\newcommand{\calX}{\mathcal{X}}

\newcommand{\calN}{\mathcal{N}}

\newcommand{\calG}{\mathcal{G}}

\newcommand{\RR}{\mathbb{R}}

\newcommand{\gray}[1]{{\color{gray}#1}}
\newcommand{\red}[1]{{\color{red}#1}}

\usepackage{pifont}
\title{FedPass: Privacy-Preserving Vertical Federated Deep Learning \\ with Adaptive Obfuscation}

\author{
    Author Name
    \affiliations
    Affiliation
    \emails
    email@example.com
}

\author{
Hanlin Gu$^{*1}$
\and
Jiahuan Luo$^{*1}$\and
Yan Kang$^{1}$\and
Lixin Fan$^{\dagger 1}$\And
Qiang Yang$^{1,2}$
\affiliations
$^1$Webank, $^2$HKUST
\emails
\{allengu,jiahuanluo,yangkang,lixinfan\}@webank.com,
qyang@cse.ust.hk
}

\begin{document}

\maketitle
\def\thefootnote{*}\footnotetext{These authors contributed equally to this work}\def\thefootnote{\arabic{footnote}}
\def\thefootnote{$\dagger$}\footnotetext{Corresponding author}\def\thefootnote{\arabic{footnote}}
\begin{abstract}
Vertical federated learning (VFL) allows an active party with labeled feature to leverage auxiliary features from the passive parties to improve model performance. Concerns about the private feature and label leakage in both the training and inference phases of VFL have drawn wide research attention. In this paper, we propose a general privacy-preserving vertical federated deep learning framework called FedPass, which leverages adaptive obfuscation to protect the feature and label simultaneously. Strong privacy-preserving capabilities about private features and labels are theoretically proved (in Theorems \ref{thm:thm1} and \ref{thm2}). 
Extensive experimental result s with different datasets and network architectures also justify the superiority of FedPass against existing methods in light of its near-optimal trade-off between privacy and model performance.

\end{abstract}

\section{Introduction}

Vertical federated learning (VFL)~\cite{yang2019federated} allows multiple organizations to exploit in a privacy-preserving manner their private datasets, which may have some sample IDs in common but are significantly different from each other in terms of \textit{features}. VFL found a great deal of successful applications, especially in collaborations between banks, healthcare institutes, e-commerce platforms, etc.~\cite{yang2019federated,li2020review}.

Despite of these successful VFL applications, privacy risks ascribed to certain corner cases were reported, e.g., in
~\cite{jin2021cafe,fu2022label,he2019model}. 
On one hand, \textit{active parties} in VFL are concerned by the risk of leaking \textit{labels} to {passive parties}. On the other hand, \textit{passive parties} are keen to protect their \textit{private features} from being reconstructed by the active party. 
To this end, a variety of privacy defense mechanisms include adding noise \cite{fu2022label,liu2021defending}, gradient discretization \cite{dryden2016communication},
gradient sparsification \cite{aji2017sparse}, gradient compression \cite{lin2018deep}, and mixup \cite{huang2020instahide,zhang2018mixup} has been proposed to boost the \textit{privacy-preserving capability} of VFL. Nevertheless, as shown by detailed analysis and empirical study in this paper (see Sect. \ref{sec:exp}), all the aforementioned defense mechanisms suffer from deteriorated model performance to a certain extent. 
In this paper, we analyze the root cause of compromised model performance and propose an effective privacy-preserving method, which not only provides a strong privacy guarantee under various privacy attacks but also maintains unbending model performance for a variety of experimental settings. 





Existing privacy defense mechanisms \cite{fu2022label,liu2020secure,dryden2016communication,aji2017sparse,lin2018deep,huang2020instahide} can be essentially viewed as an \textbf{obfuscation mechanism} that provides privacy guarantees by adopting an obfuscation function $g(\cdot)$\footnote{For the sake of brevity, we omit some detailed notations, which are elaborated on in Sect. \ref{sec:notation}} acting on private features $x$ or labels $y$, illustrated as follows: 
\begin{flalign}\label{eq:fix-obfuscation}
        x \stackrel{g(\cdot)}
        \longrightarrow g(x) \stackrel{G_\theta}{\longrightarrow} H \stackrel{F_\omega}\longrightarrow \ell \longleftarrow g(y)
\stackrel{g(\cdot)}\longleftarrow 
y, 
\end{flalign}
in which $H$ is the forward embedding that the passive party transfers to the active party; $\ell$ is loss; $G$ and $F$ represent the passive and active model parameterized by $\theta$ and $\omega$, respectively.
Note that in Eq. \eqref{eq:fix-obfuscation}, the strength of privacy-preserving capability is prescribed by the \textit{extent of obfuscation} (a.k.a. \textit{distortion}) via a fixed hyper-parameter. As investigated in \cite{zhang2022trading,kang2022framework}, significant obfuscations inevitably bring about the loss of information in $g(x)$ and $g(y)$ and thus lead to the loss of model performance (empirical study in Sect. \ref{sec:exp}). We argue that a fixed obfuscation strategy does not consider the dynamic learning process and constitutes the root cause of model performance degradation.  

\begin{figure}[htbp]
\centering
\includegraphics[width=0.45\textwidth]{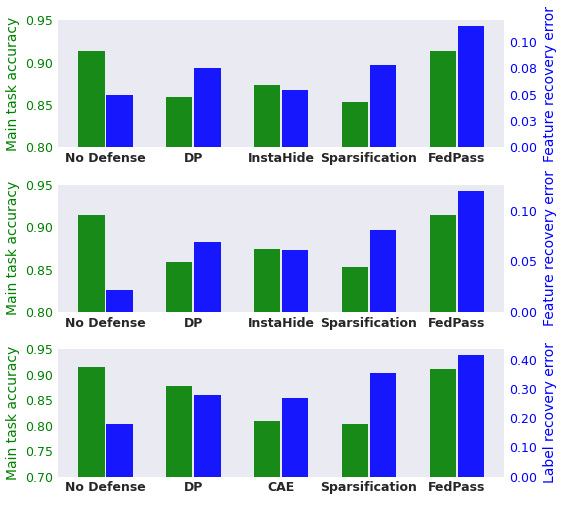}
\vspace{-1em}
\caption{Comparison of FedPass with baseline defense methods against \textbf{Model Inversion} attack \protect\cite{he2019model} (the first line), \textbf{CAFE} attack \protect\cite{jin2021cafe} (the second line), and \textbf{Model Completion} attack \protect\cite{fu2022label} (the third line) in terms of their main task accuracy (the higher the better, green) and data (feature or label) recovery error (the higher the better, blue) on ResNet-CIFAR10.}
\label{fig:bar_brief}
\end{figure}

As a remedy to shortcomings of the fixed obfuscation,  we propose to \textbf{adapt obfuscation function} $g_{\theta}$ and $g_{\omega}$ during the learning of model $F_\theta$ and $G_\omega$, such that the obfuscation itself is also optimized during the learning stage. 
That is to say, the learning of the obfuscation function also aims to preserve model performance by tweaking model parameters $\theta$ and $\omega$:
\vspace{-0.4em}
\begin{flalign} \label{eq:adapt-obsfucation}
                x \stackrel{g_\red{\theta}(\cdot)}{\longrightarrow} g_{ \textcolor{red}{\theta}}(x) \stackrel{G_\theta}{\longrightarrow} H  \stackrel{g_\red{\omega}(\cdot)}{\longrightarrow} g_{\textcolor{red}{\omega}}(H)
\stackrel{F_\omega}{\longrightarrow} \ell \longleftarrow y,
\end{flalign}
We regard this adaptive obfuscation as the gist of the proposed method. In this work, we implement the adaptive obfuscation based on the passport technique originally designed for protecting the intellectual property of deep neural networks (DNN) \cite{fan2021deepip,li2022fedipr}. More specifically, we propose to embed \textit{private passports} in both active and passive party models of VFL to effectively defend against privacy attacks. We therefore name the proposed method \textbf{FedPass},
which has three advantages: i) Private passports embedded in passive and active models prevent attackers from inferring features and labels. It is \textit{exponentially hard} 
to infer features by launching various attacks, while attackers are defeated by a \textit{non-zero recovery error} when attempting to infer private labels (see Sect. \ref{sec:analysis}).
ii) Passport-based obfuscation is learned in tandem with the optimization of model parameters, thus, preserving model performance (see Sect. \ref{sec:fedpass}). iii) The learnable obfuscation is efficient, with only minor computational costs incurred since no computationally extensive encryption operations are needed (see Sect. \ref{sec:exp}). 

As Figure \ref{fig:bar_brief} illustrated, 
FedPass achieves almost lossless main task accuracy while obtaining the largest data recovery error in defending against three privacy attacks investigated in our experiments (see  Sect. \ref{experiments} for more results).

\section{Related work} 
We review related work from three aspects, namely, \textit{vertical federated learning} (VFL) and \textit{privacy attacks in VFL}, and \textit{protection mechanism} 
\textbf{Vertical Federated Learning} A variety of VFL methods has been proposed.~\cite{hardy2017private} proposed vertical logistic regression (VLR) using homomorphic encryption (HE) to protect feature privacy.~\cite{blindFL} proposed vertical neural network that employs a hybrid privacy-preserving strategy combining HE and secret sharing (SS) to protect feature.~\cite{secureboost} proposed the SecureBoost, a VFL version of XGBoost, that leverages HE to protect the information exchanged among parties. To tackle the data deficiency issue of VFL,~\cite{kangyan2022fedcvt} combined semi-supervised learning and cross-view training to estimate missing features and labels for further training,~\cite{liu2020secure,kang2022prada} integrated transfer learning into VFL to help the target party predict labels, while~\cite{he2022hybrid} proposed a
federated hybrid self-supervised learning framework to boost the VFL model performance through self-supervised learning based on unlabeled data. 


\textbf{Privacy attacks in VFL}: There are two categories of privacy attacks in VFL: feature inference (FI) attacks and label inference (LI) attacks. FI attacks are typically launched by the active party to infer the features of a passive party. They can be tailored to shallow models such as logistic regression~\cite{hu2022vertical} and decision trees~\cite{Luo2021fi}. For deep neural networks, model inversion~\cite{he2019model} and CAFE~\cite{jin2021cafe} are two representative FI attacks that infer private features through inverting passive parties' local models. LI attacks are mounted by a passive party to infer labels owned by the active party. The literature has proposed two kinds of LI attacks: the gradient-based~\cite{oscar2022split} and the model-based~\cite{fu2022label}. The former infers labels though analyzing the norm or direction of gradients back-propagated to the passive party, while the latter infers labels based on a pre-trained attacking model.

\textbf{Defense mechanism}
We divide defense mechanisms applied to VFL into three categories: Protect features, labels and model parameters (or gradients).
Specifically, most existing defense mechanisms typically apply to either model parameters or gradients to protect private data (features or labels). General defense mechanisms such as differential privacy~\cite{abadi2016deep} and sparsification~\cite{fu2022label,lin2018deep} can be used to defend against both the feature and label inference attacks by distorting (i.e., adding noise or compressing) model parameters or gradients. Specialized defense mechanisms such as MARVELL~\cite{oscar2022split} and Max-Norm~\cite{oscar2022split} are tailored to thwart label inference attacks by adding noise to gradients. InstaHide~\cite{huang2020instahide} and Confusional AutoEncoder~\cite{zou2022defending} are two representative defense mechanisms that encode private data directly to protect data privacy. Cryptography-based defense mechanisms such as HE \cite{hardy2017private} and MPC \cite{gascon2016secure} can protect both features and labels, but they impose a huge burden on computation and communication, especially for deep neural networks. 

\begin{table*}[htbp] 
\center
\begin{tabular}{ c c c c c } 
\hline
  Threat model & Adversary & Attacking Target & Attacking Method & Adversary's Knowledge  \\ 
   \hline \hline
   \multirow{3}{*}{Semi-honest} & A passive party & \textit{Labels} owned by the active party & Model Completion & A few labeled samples \\
   \cline{2-5}
   & \multirow{2}{*}{The active party} & \textit{Features} owned by a passive party & Model Inversion & Some labeled samples \\ 
   &&\textit{Features} owned by a passive party & CAFE & Passive models \\ \hline
\end{tabular}
\vspace{-0.6em}
\caption{Threat model we consider in this work.}
\label{tab_threat_model}
\end{table*}

\section{The Proposed Method}



We introduce the VFL setting and threat models in Sect. \ref{sec:notation} and Sect. \ref{sec:threat-model}, followed by elaboration on the proposed adaptive obfuscation framework, FedPass, in Sect. \ref{sec:fedpass}. We analyze the privacy preserving capability of FedPass in Sect. \ref{sec:analysis}.
\begin{figure}[!t]
\centering
\includegraphics[width=0.44\textwidth]{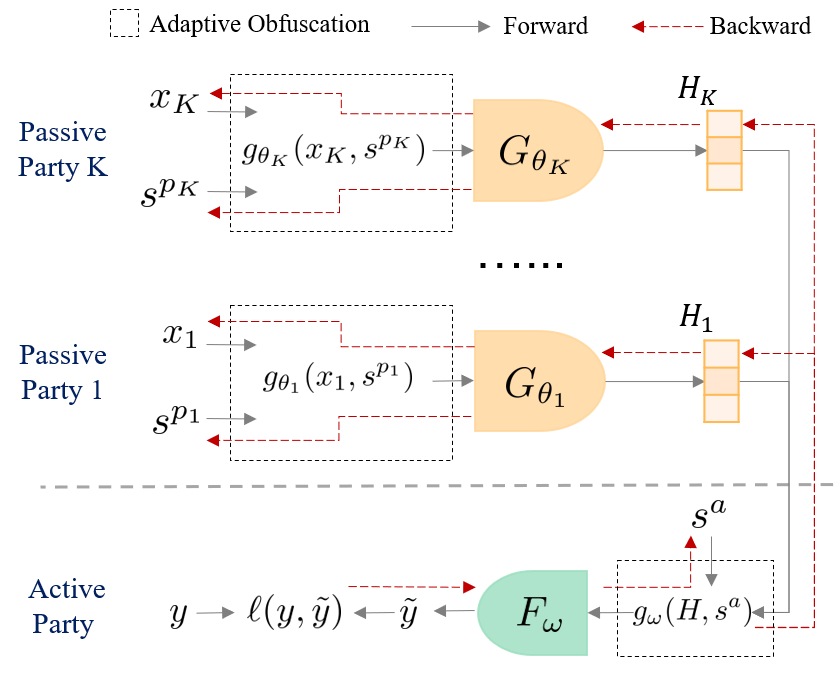}
\vspace{-3pt}
\caption{Overview of a VFL setting, in which multiple passive parties and one active party collaboratively train a VFL model, where passive parties only have the private features $x$, whereas the active party has private labels $y$. Both the active party and the passive party adopt the adaptive obfuscation by inserting passport into their models to protect features and labels.}
\label{fig:Fedpass}
\end{figure}

\subsection{Vertical Federated Learning Setting} \label{sec:notation}
We assume that a vertical federated learning setting consists of one active party $P_0$ and $K$ passive parties $\{P_1, \cdots, P_K\}$  who collaboratively train a VFL model $\Theta=(\theta, \omega)$ to optimize the following objective: 
\begin{equation}\label{eq:loss-VFL}
\begin{split}
        \min_{\omega, \theta_1, \cdots, \theta_K} &\frac{1}{n}\sum_{i=1}^n\ell(F_{\omega} \circ (G_{\theta_1}(x_{1,i}),G_{\theta_2}(x_{2,i}), \\
        & \cdots,G_{\theta_K}(x_{K,i})), y_{i}),
\end{split}
\end{equation}

in which Party $P_k$ owns features $\calD_k = (x_{k,1}, \cdots, x_{k,n}) \in \mathcal{X}_k$ and the passive model $G_{\theta_k}$, the active party owns the labels $y \in \mathcal{Y}$ and active model $F_\omega$, $\mathcal{X}_k$ and $\mathcal{Y}$ are the feature space of party $P_k$ and the label space respectively. Each passive party $k$ transfers its forward embedding $H_k$ to the active party to compute the loss. The active model $F_\omega$ and passive models $G_{\theta_k},k \in \{1,\dots,K\}$ are trained based on backward gradients (See Figure \ref{fig:Fedpass} for illustration). Note that, before training, all parties leverage Private Set Intersection (PSI) protocols to align data records with the same IDs. 



\subsection{Threat Model}\label{sec:threat-model}
We assume all participating parties are \textit{semi-honest} and do not collude with each other. An adversary (i.e., the attacker) $P_k,k=0,\cdots,K$ faithfully executes the training protocol but may launch privacy attacks to infer the private data (features or labels) for other parties. 


We consider two types of threat models: \romannumeral1) The active party wants to reconstruct the private features of a passive party through the model inversion attack \cite{he2019model} or CAFE attack~\cite{jin2021cafe}. \romannumeral2) A passive party wants to infer the private labels of the active party through the model completion attack~\cite{fu2022label}. 
Table \ref{tab_threat_model} summarizes threat models considered in this work.


\subsection{FedPass} \label{sec:fedpass}

This section illustrates two critical steps of the proposed FedPass, i) embedding private passports to adapt obfuscation; ii) generating passports randomly to improve privacy preserving capability.

\begin{figure}[!ht]
\centering
\includegraphics[width=0.48\textwidth]{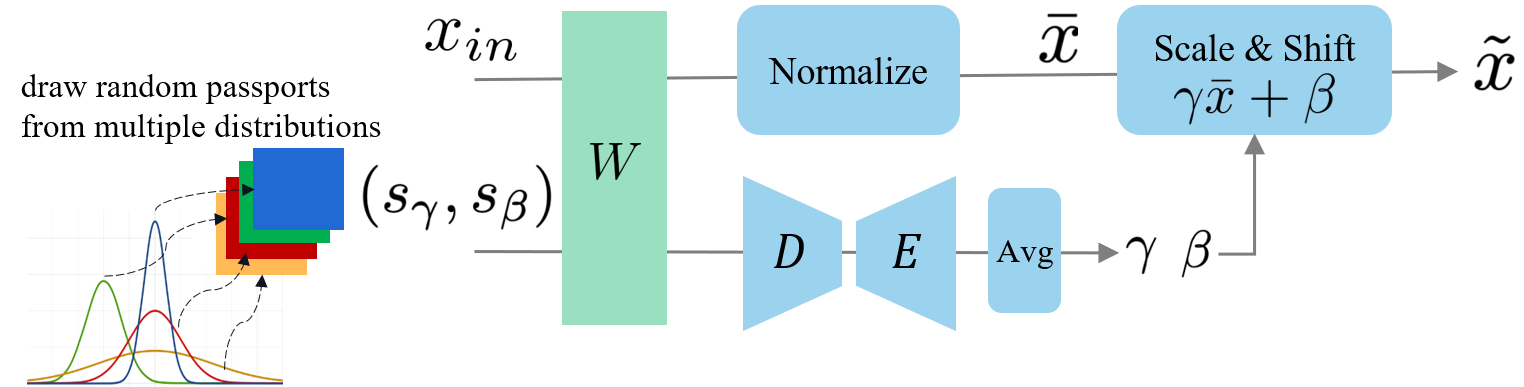}
 \vspace{-1em}
\caption{Adaptive obfuscation ($g(\cdot)$). We implemented $g(\cdot)$ by inserting a passport layer into a normal neural network layer.}
\label{fig:passport_layer}
\vspace{-3pt}
\end{figure}


\subsubsection{Embedding Private Passports}
In this work, we adopt the DNN passport technique proposed by \cite{fan2019rethinking,fan2021deepip} as an implementation for the adaptive obfuscation framework of FedPass. 
Specifically, the adaptive obfuscation is determined as follows:
\begin{equation}\label{eq:pst1}
\begin{split}
        g_{W}(x_{in}, s) = &
        \gamma(Wx_{in}) + \beta, \\
         \gamma=&\text{Avg}\Big( D\big(E(Ws_\gamma)\big)\Big)\\
        \beta = &\text{Avg}\Big( D\big(E(Ws_\beta)\big)\Big)
\end{split}
\end{equation}
where $W$ denotes the model parameters of the neural network layer for inserting passports, $x_{in}$ is the input fed to $W$, $\gamma$ and $\beta$ are the scale factor and the bias term. Note that the determination of the crucial parameters $\gamma$ and $\beta$ involves the model parameter $W$ with private passports $s_\gamma$ and $s_\beta$, followed by a autoencoder (Encoder $E$ and Decoder $D$ with parameters $W'$) 
and a average pooling operation Avg($\cdot$). Learning adaptive obfuscation formulated in Eq.(\ref{eq:pst1}) brings about two desired properties: 
\begin{itemize}
    \item Passport-based parameters $\gamma$ and $\beta$ provide strong privacy guarantee (refer to Sect.~\ref{sec:analysis}): without knowing passports it is exponentially hard for attacker to infer layer input $x_{in}$ from layer output $g_{W}(x_{in}, s)$, because attacker have no access to $\gamma$ and $\beta$ (see Theorem \ref{thm:thm1}).
    \item Learning adaptive obfuscation formulated in Eq. (\ref{eq:pst1}) optimizes the model parameter $W$ through three backpropagation paths via $\beta, \gamma, W$, respectively, which helps preserve model performance. This is essentially equivalent to adapting the obfuscation (parameterized by $\gamma$ and $\beta$) to the model parameter $W$ (more explanations in Appendix C); This adaptive obfuscation scheme offers superior model performance compared to fixed obfuscation schemes (see Sect. \ref{sec:exp}). 

\end{itemize}


The training procedure of FedPass in Vertical Federated Learning is illustrated as follows (described in Algorithm \ref{alg:aof-vfl}):
\begin{enumerate}
    \item Each passive party $k$ applies the adaptive obfuscation to its private features with its private passports $s^{p_k}$ and then sends the forward embedding $H_k$ to the active party (line 3-9 of Algo. \ref{alg:aof-vfl});
    \item The active party sums over all $H_k, k \in \{1,\dots, K\}$ as $H$, and applies the adaptive obfuscation to $H$ with its private passports $s^a$, generating $\Tilde{H}$. Then, the active party computes the loss $\Tilde{\ell}$ and updates its model through back-propagation. Next, the active party computes gradients $\nabla_{H_{k}}\Tilde{\ell}$ for each passive party $k$ and sends $\nabla_{H_{k}}\Tilde{\ell}$ to passive party $k$ (line 10-19 of Algo. \ref{alg:aof-vfl});
    \item Each passive party $k$ updates its model $\theta_k $ according to $\nabla_{H_k}\Tilde{\ell}$ (line 20-22 of Algo. \ref{alg:aof-vfl}). 
\end{enumerate} 

The three steps iterate until the performance of the joint model does not improve.

\subsubsection{Random Passport Generation}
How passports are generated is crucial in protecting data privacy. 
Specifically, when the passports are embedded in a convolution layer or linear layer with $c$ channels\footnote{For the convolution layer, the passport $s\in \RR^{c\times h_1 \times h_2}$, where $c$ is channel number, $h_1$ and $h_2$ are height and width; for the linear layer, $s\in \RR^{ c\times h_1}$, where $c$ is channel number, $h_1$ is height.}, 
for each channel $j\in[c]$, 
the passport $s{(j)}$ (the $j_{th}$ element of vector $s$) is randomly generated as follows:
\begin{equation}\label{eq:sample-pst}
    s{(j)} \sim \calN(\mu_j, \sigma^2), \quad
    \mu_j \in \calU(-N, 0),
\end{equation}
where all $\mu_j, j=1,\cdots, c$ are different from each other, $\sigma^2$ is the variance of Gaussian distribution and $N$ is the \textit{passport range},  which are two crucial parameters of FedPass. The strong privacy-preserving capabilities rooted in such a random passport generation strategy are justified by theoretical analysis in  Theorems \ref{thm:thm1} and \ref{thm2} as well as experiment results in Sect. \ref{sec:exp}. 


\begin{algorithm}[!ht]\vspace{-3pt}
	\caption{FedPass}
	\begin{algorithmic}[1]
	   \Statex \textbf{Input:} Communication rounds $T$, Passive parties number $K$, learning rate $\eta$, batch size $b$, the passport range and variance $\{N^a, \sigma^a\}$ and $\{N^{p_k}, \sigma^{p_k}\}$ for the active party and passive party $k$ respectively, the feature dataset $\calD_k = (x_{k,1},\cdots, x_{k,n_k})$ owned by passive party $k$, the aligned label $y = (y_1, \cdots, y_{n_0})$ owned by the active party.

        \Statex \textbf{Output:} Model parameters $\theta_1, \cdots, \theta_K, \omega$ \vspace{4pt}
	    
	\State Initialize model weights $\theta_1, \cdots, \theta_K, \omega$.
     \For{$t$ in communication round $T$} \vspace{2pt}
        \State \gray{$\triangleright$ \textit{Passive parties perform:}}
         \For{Passive Party $k$ in $\{1,\dots,K\}$}: 
         \State Sample a batch $B_k=(x_{k,1},\cdots, x_{k,b})$ from the dataset $\calD_k$
         \State Sample the passport tuple $s^{p_k} = (s^{p_k}_\gamma, s^{p_k}_\beta)$ according to Eq. \eqref{eq:sample-pst} and $N^{p_k}, \sigma^{p_k}$
          \State Compute $\Tilde{B_k} = g_{\theta_k}( B_k,s^{p_k})$
\State   Compute  $H_{k} \gets G_{\theta_k}(\Tilde{B_k})$
\State  Send  $H_{k}$  to the active party
\EndFor
\State \gray{$\triangleright$ \textit{The active party performs:}}
\State Obtain the label $y$ matched with $\{B_k\}_{k=1}^K$.
\State $H = \sum_{k=1}^KH_k$ 
\State Sample the passport tuple $s^a=(s^{a}_\gamma, s^{a}_\beta)$ via Eq. \eqref{eq:sample-pst} and $N^a, \sigma^a$ 
\State Compute $\Tilde{H} = g_{\omega}(H,s^a)$
\State Sample a batch of the label $Y = (y_1,\cdots, y_b)$
\State Compute cross-entropy loss:
$\Tilde{\ell} = \ell(F_{\omega}(\Tilde{H}),Y)$
\State Update the active model as: $\omega = \omega - \eta \nabla_\omega\Tilde{\ell}$
   \For {$k$ in $\{1,\dots,K\}$}:   
   \State Compute and send  $\nabla_{{H_{k}}}\tilde{\ell}$ to each passive party $k$
   \EndFor		
\State \gray{$\triangleright$ \textit{Passive parties perform:}}
   \For{Passive Party $k \in \{1,\dots,K\}$}: 
      \State  Update $\theta_k$ by $\theta_k = \theta_k - \eta [\nabla_{H_k}\tilde{\ell}] [\nabla_{\theta_k}H_k]$
\EndFor
\EndFor
\Return $\theta_1, \cdots, \theta_K, \omega$
	\end{algorithmic}\label{alg:aof-vfl}
\end{algorithm}

\begin{algorithm}[!ht]
\caption{Adaptive Obfuscation ($g(\cdot)$)}
\begin{algorithmic}[1]
  \Statex \textbf{Input:} Model parameters $W$ of the neural network layer for inserting passports; the input $x_{in}$ to which the adaptive obfuscation applies to; passport keys  $s = (s_\gamma, s_\beta)$.
  \Statex \textbf{Output:} The obfuscated version of the input.
  \State Compute $
         \gamma=\text{Avg}\left( D\big(E(W* s_\gamma)\big)\right)$
\State Compute $ \beta = \text{Avg}\left( D\big(E(W* s_\beta)\big)\right)$ \\
\Return $\gamma(W* x_{in}) + \beta$
	\end{algorithmic}
	\label{alg:AO}
\end{algorithm}

\subsection{Privacy-Preserving Capability of FedPass}
\label{sec:analysis}
We investigate the privacy-preserving capability of FedPass against feature reconstruction attack and label inference attack. 
Note that we conduct the privacy analysis with linear regression models, for the sake of brevity. Proofs are deferred to Appendix D.
\begin{definition} \label{def:SplitFed}
Define the forward function of the passive model $G$ and the active model $F$: 
\begin{itemize}
    \item For passive layer: $H = G(x) =  W_p s_\gamma^p \cdot W_p x + W_p s_\beta^p$.
    \item For active layer: $y = F(H) =  W_a s_\gamma^a \cdot W_a  H + W_a s_\beta^a$.
\end{itemize}
where $W_p$, $W_a$ are 2D matrices of the passive and active models; $\cdot$ denotes the inner product, $ s_\gamma^p,  s_\beta^p$ are passports of the passive party, $ s_\gamma^a,  s_\beta^a$ are passports of the active party.
\end{definition}

\subsubsection{Hardness of feature Restoration with FedPass}
Consider the white-box Model Inversion (MI) attack (i.e., model inversion step in CAFE \cite{jin2021cafe,he2019model}) that aims to inverse the model $W_p$ to recover features $\hat{x}$ approximating original features $x$. In this case, the attacker (i.e., the active party) knows the passive model parameters $W_p$, forward embedding $H$ and the way of embedding passport, but does not know the passport.

\begin{theorem}\label{thm:thm1}
    Suppose the passive party protects features $x$ by inserting the $s_{\beta}^p$. The probability of recovering features by the attacker via white-box MI attack is at most $\frac{\pi^{m/2}\epsilon^m}{\Gamma(1+m/2)N^m}$ such that the recovering error is less than $\epsilon$, i.e., $\|x-\hat{x}\|_2\leq \epsilon$,
\end{theorem}
where $m$ denotes the dimension of the passport via flattening, $N$ denotes the passport range formulated in Eq. \eqref{eq:sample-pst} and $\Gamma(\cdot)$ denotes the Gamma distribution.
Theorem \ref{thm:thm1} demonstrates that the attacker's probability of recovering features within error $\epsilon$ is exponentially small in the dimension of passport size $m$. The successful recovering probability is inversely proportional to the passport range $N$ to the power of $m$.
\subsubsection{Hardness of Label Recovery with FedPass}
Consider the passive model competition attack \cite{fu2022label} that aims to recover labels owned by the active party. The attacker (i.e., the passive party) leverages a small auxiliary labeled dataset $\{x_i, y_i\}_{i=1}^{n_a}$ belonging to the original training data to train the attack model $W_{att}$, and then infer labels for the test data. Note that the attacker knows the trained passive model $G$ and forward embedding $H_i = G(x_i)$. Therefore, they optimizes the attack model $W_{att}$ by minimizing $\ell = \sum_{i=1}^{n_a}\|W_{att}H_i-y_i\|_2$.
\begin{assumption}\label{assum1}
Suppose the original main algorithm of VFL is convergent. For the attack model, we assume the error of the optimized attack model $W^*_{att}$ on test data $\tilde{\ell}_t$ is larger than that of the auxiliary labeled dataset $\Tilde{\ell}_a$.
\end{assumption}
\begin{theorem} \label{thm2}
Suppose the active party protect $y$ by embedding $s^a_\gamma$, and adversaries aim to recover labels on the test data with the error $\Tilde{\ell}_t$ satisfying:
    \begin{equation}
        \Tilde{\ell}_t \geq \min_{W_{att}} \sum_{i=1}^{n_a}\|(W_{att}- T_i)H_i\|_2,
    \end{equation}
    where $T_i =diag(W_as_{\gamma,i}^a) W_a$ and $s_{\gamma,i}^a$ is the passport for the label $y_i$ embedded in the active model. Moreover, if $H_{i_1} = H_{i_2} = H$ for any $1\leq i_1,i_2 \leq n_a$, then
    \begin{equation} \label{eq:protecty}
        \Tilde{\ell}_t \geq \frac{1}{(n_a-1)}\sum_{1\leq i_1<i_2\leq n_a}\|(T_{i_1}-T_{i_2})H\|_2
    \end{equation}
\end{theorem}
\begin{prop}\label{prop1}
Since passports are randomly generated and $W_a$ and $H$ are fixed, if the $W_a = I, H=\Vec{1}$, then it follows that:
\begin{equation}
    \Tilde{\ell}_t \geq \frac{1}{(n_a-1)}\sum_{1\leq i_1<i_2\leq n_a}\|s_{\gamma,i_1}^a-s_{\gamma,i_2}^a\|_2)
\end{equation}
\end{prop}

Theorem \ref{thm2} and Proposition \ref{prop1} show that the label recovery error $\Tilde{\ell}_t$ has a lower bound,
which deserves further explanations. First, when passports are randomly generated for all data, i.e., $s_{\gamma,i_1}^a \neq s_{\gamma,i_2}^a$, then a non-zero label recovery error is guaranteed no matter how adversaries attempt to minimize it. The recovery error thus acts as a protective random noise imposed on true labels. Second, the magnitude of the recovery error monotonically increases with the variance $\sigma^2$ of the Gaussian distribution passports sample from (in Eq. \eqref{eq:sample-pst}), which is a crucial parameter to control privacy-preserving capability (see Experiment results in Appendix B) are in accordance with Theorem \ref{thm2}. 
Third, it is worth noting that the lower bound is based on the training error of the auxiliary data used by adversaries to launch PMC attacks. Given possible discrepancies between the auxiliary data and private labels, e.g., in terms of distributions and the number of dataset samples, the actual recovery error of private labels can be much larger than the lower bound. Again, this strong protection is observed in experiments (see Sect. \ref{sec:exp}). 


\section{Experiment}\label{experiments}
We present empirical studies of FedPass in defending against feature reconstruction attack and label inference attack.

\subsection{Experiment setting}

\subsubsection{Models \& Datasets \& VFL Setting} 
We conduct experiments on three datasets:
\textit{MNIST} \cite{lecun2010mnist}, \textit{CIFAR10} \cite{krizhevsky2014cifar} and ModelNet \cite{wu20153d}. We adopt LeNet \cite{lecun1998gradient} for conducting experiments on MNIST and adopt \textit{AlexNet} \cite{NIPS2012_c399862d} and \textit{ResNet18} \cite{he2016deep} on CIFAR10. For each dataset, passive party only provides private data while active party only provides labels. 


We simulate a VFL scenario by splitting a neural network into a bottom model and a top model and assigning the bottom model to each passive party and the top model to the active party. Table \ref{table:models} summarizes our VFL scenarios (see details of the experimental setting in Appendix A).


\begin{table}[!h]
	\centering
	\footnotesize
	\begin{tabular}{c||c|c|c|c}
	        \hline
             \\[-1em]
		\shortstack{Dataset \\ Name} & \shortstack{Model \\ Name}  & \shortstack{Model of \\ Passive Party } & \shortstack{Model of \\ Active Party}   & \# P \\
         \hline
         \hline
          \\[-1em]
            MNIST &  LeNet & 2 Conv  &  3 FC  & 2 \\
		\hline
          \\[-1em]
		CIFAR10 & AlexNet & 5 Conv  &  1 FC & 2  \\
		\hline
          \\[-1em]
		CIFAR10 & ResNet18 & 17 Conv & 1 FC & 2\\
     \hline
     \\[-1em]
     	ModelNet\tablefootnote{We conduct experiments with 7 parties on ModelNet in Appendix B due to the limitation of space.} & LeNet & 2 Conv  &  3 FC  & 7 \\
		\hline
         
	\end{tabular}
 \vspace{-0.6em}
	\caption{Models for evaluation. \# P denotes the number of parties. FC: fully-connected layer. Conv: convolution layer. 
	}
\label{table:models}
\vspace{-3pt}
\end{table}

\begin{figure*}[!h]
\vspace{-3pt}
	\centering
      		\begin{subfigure}{0.3\textwidth}
  		 	\includegraphics[width=1\textwidth]{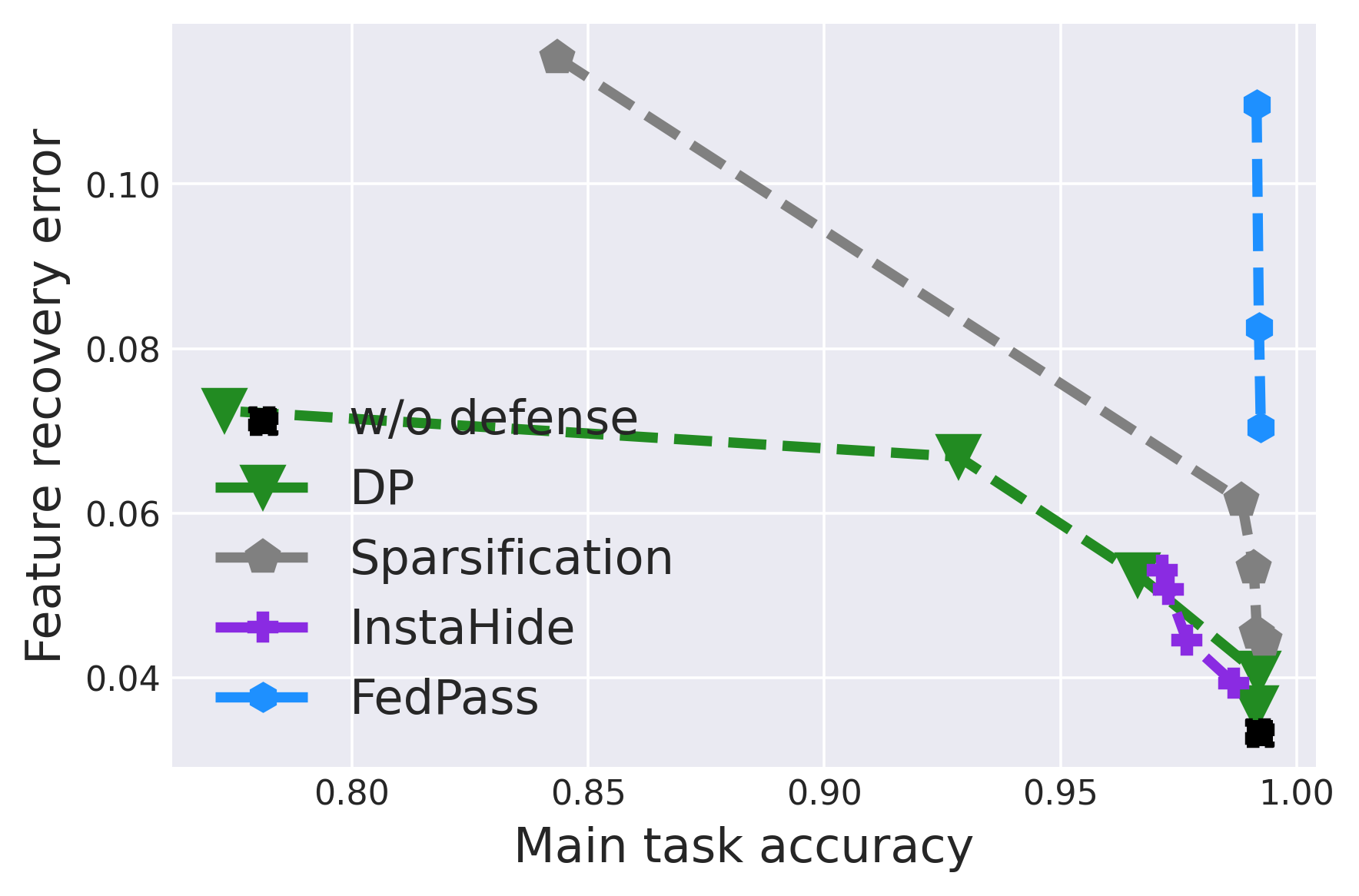}
      \subcaption{LeNet-MNIST}
    		\end{subfigure}
    	\begin{subfigure}{0.3\textwidth}
  		 	\includegraphics[width=1\textwidth]{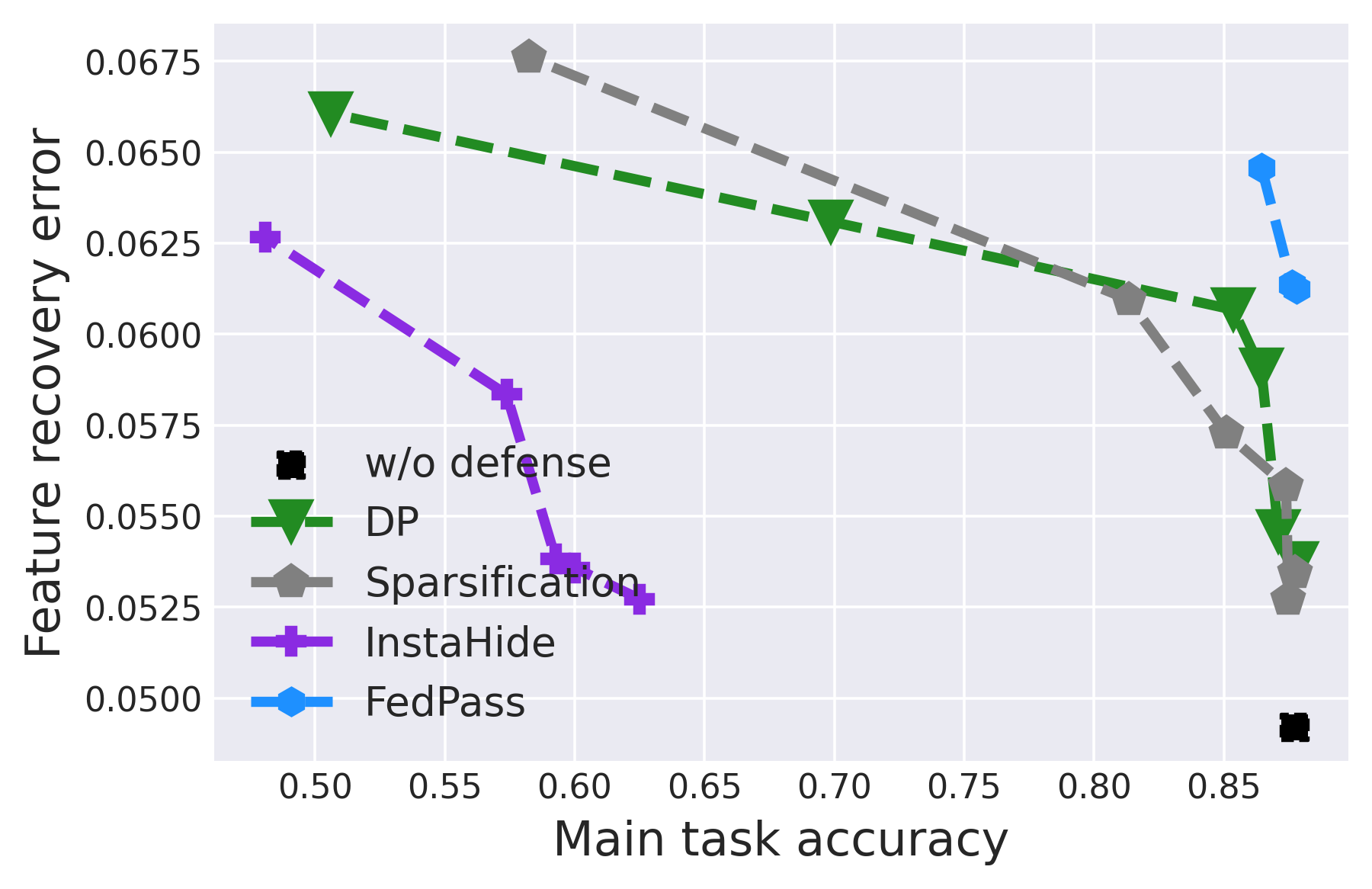}
            \subcaption{AlexNet-CIFAR10}
    		\end{subfigure}
   \begin{subfigure}{0.3\textwidth}
			\includegraphics[width=1\textwidth]{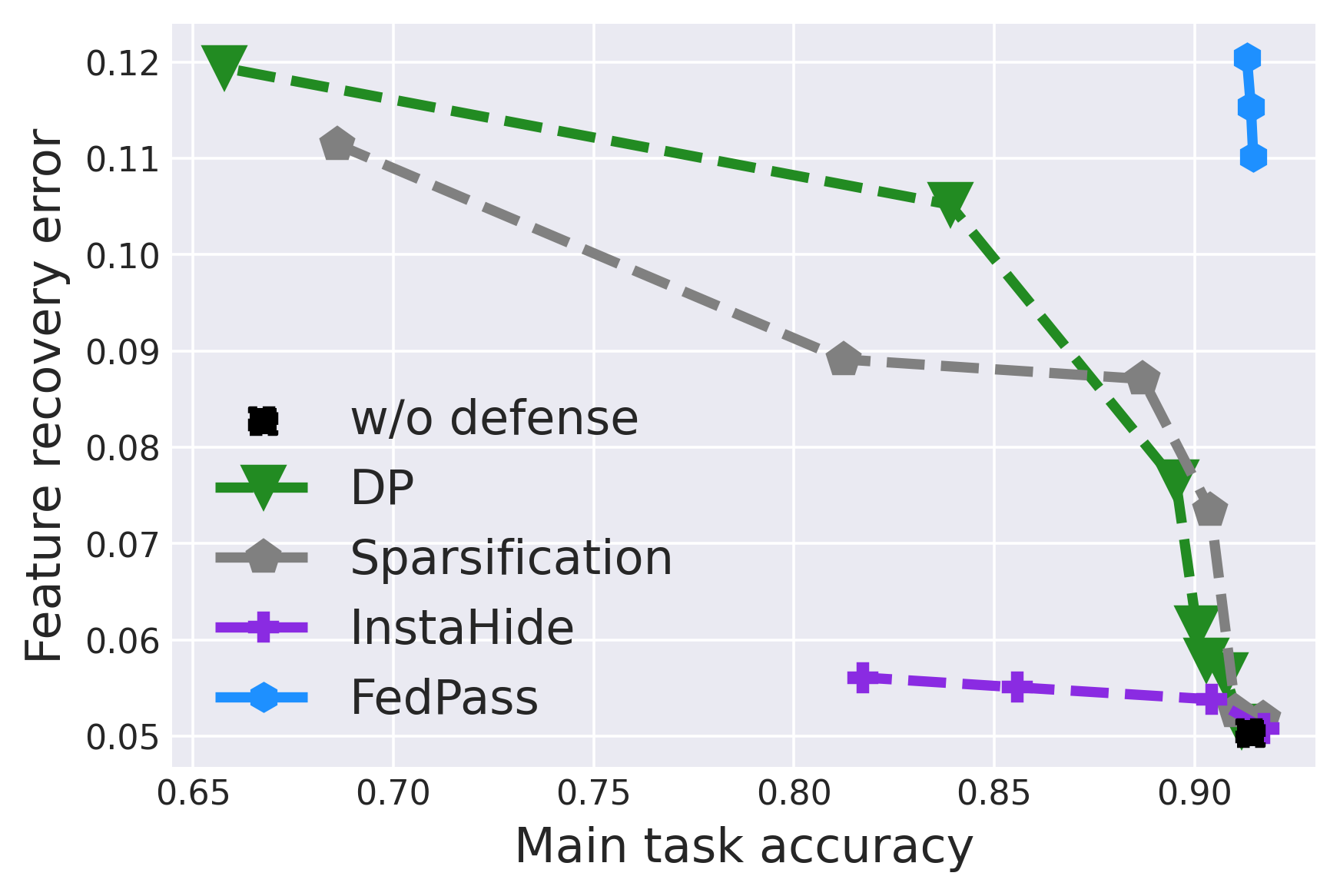}
         \subcaption{ResNet-CIFAR10}
		\end{subfigure}

      		\begin{subfigure}{0.3\textwidth}
  		 	\includegraphics[width=1\textwidth]{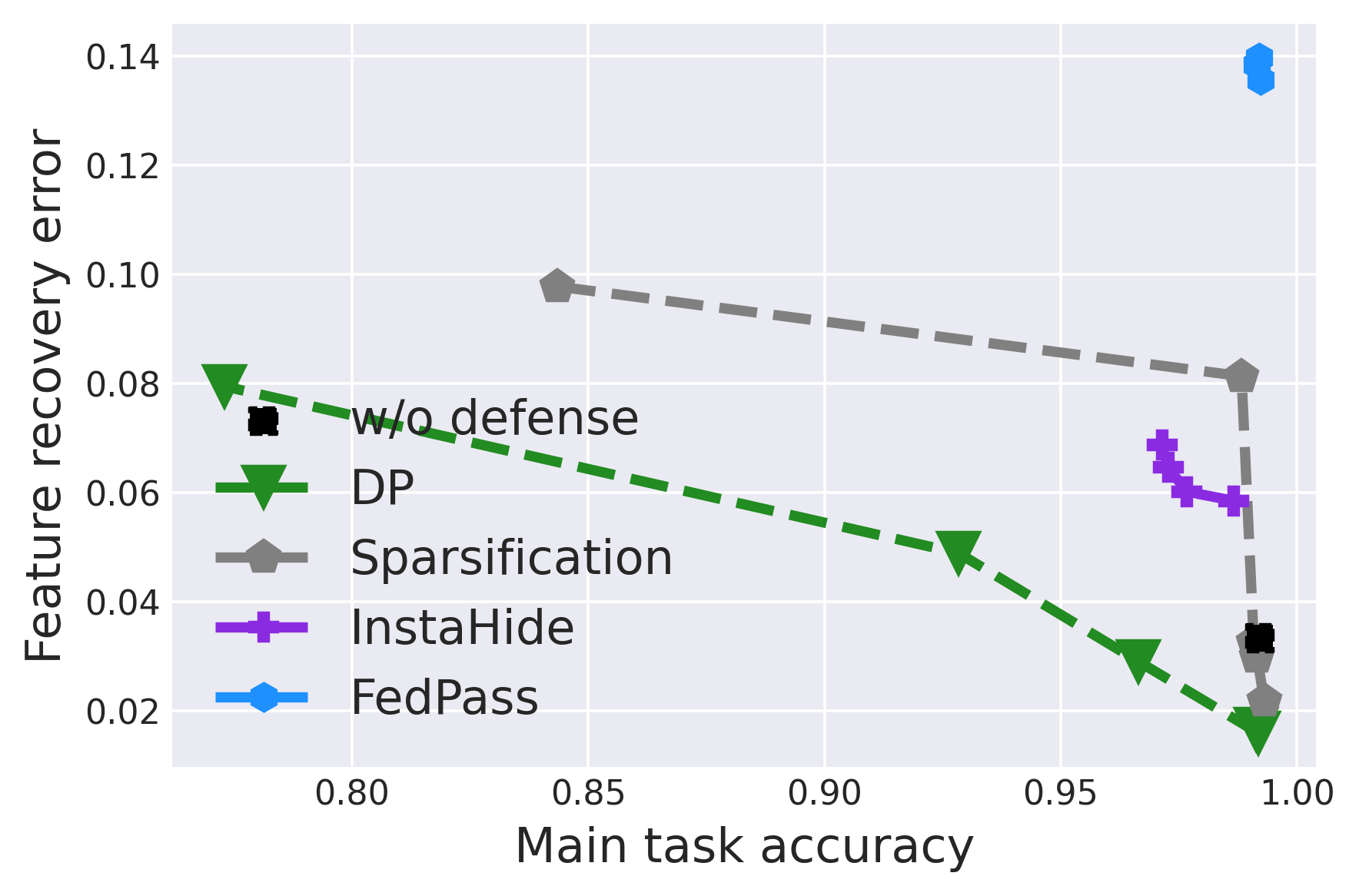}
      \subcaption{LeNet-MNIST}
    		\end{subfigure}
    	\begin{subfigure}{0.3\textwidth}
  		 	\includegraphics[width=1\textwidth]{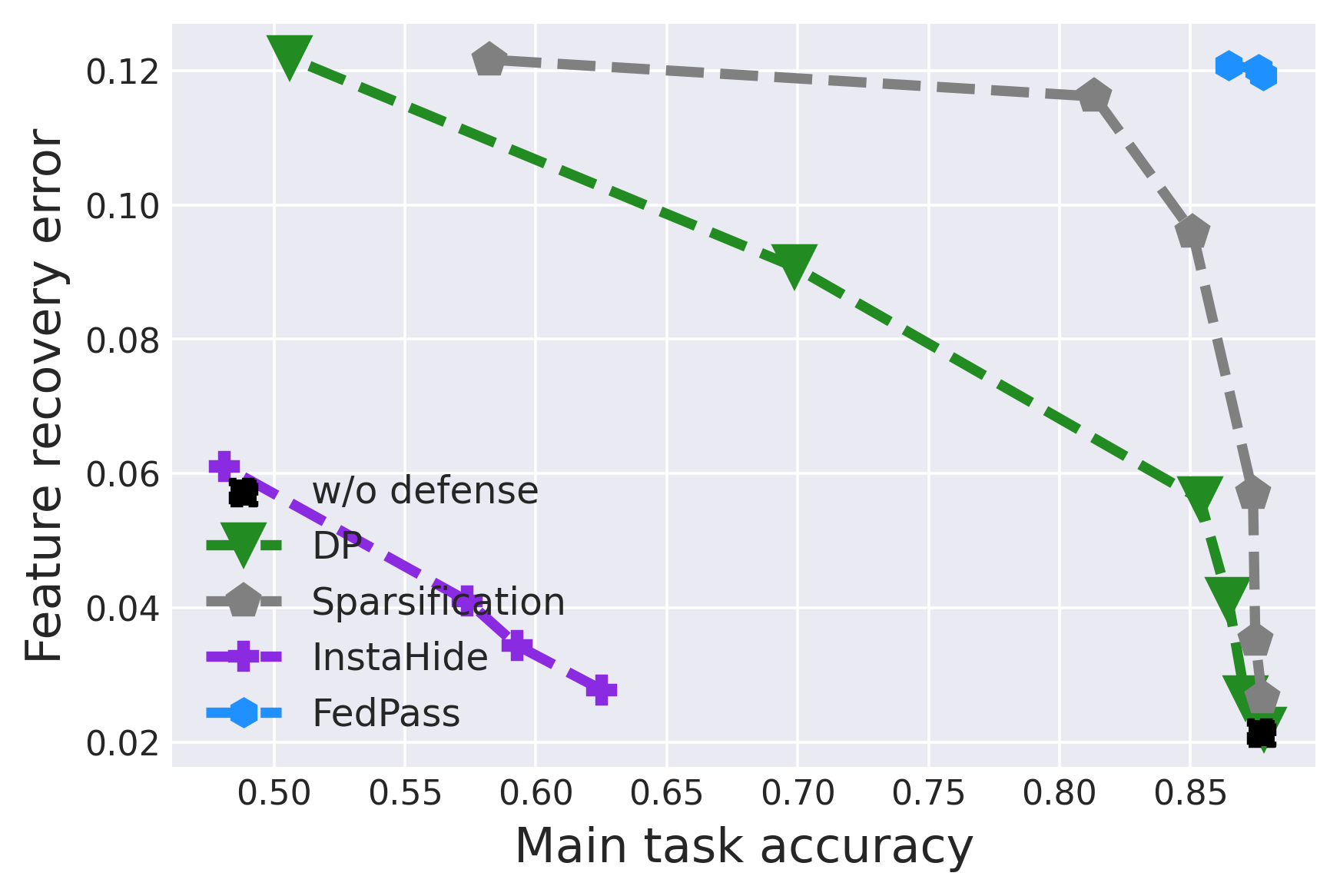}
            \subcaption{AlexNet-CIFAR10}
    		\end{subfigure}
   \begin{subfigure}{0.3\textwidth}
			\includegraphics[width=1\textwidth]{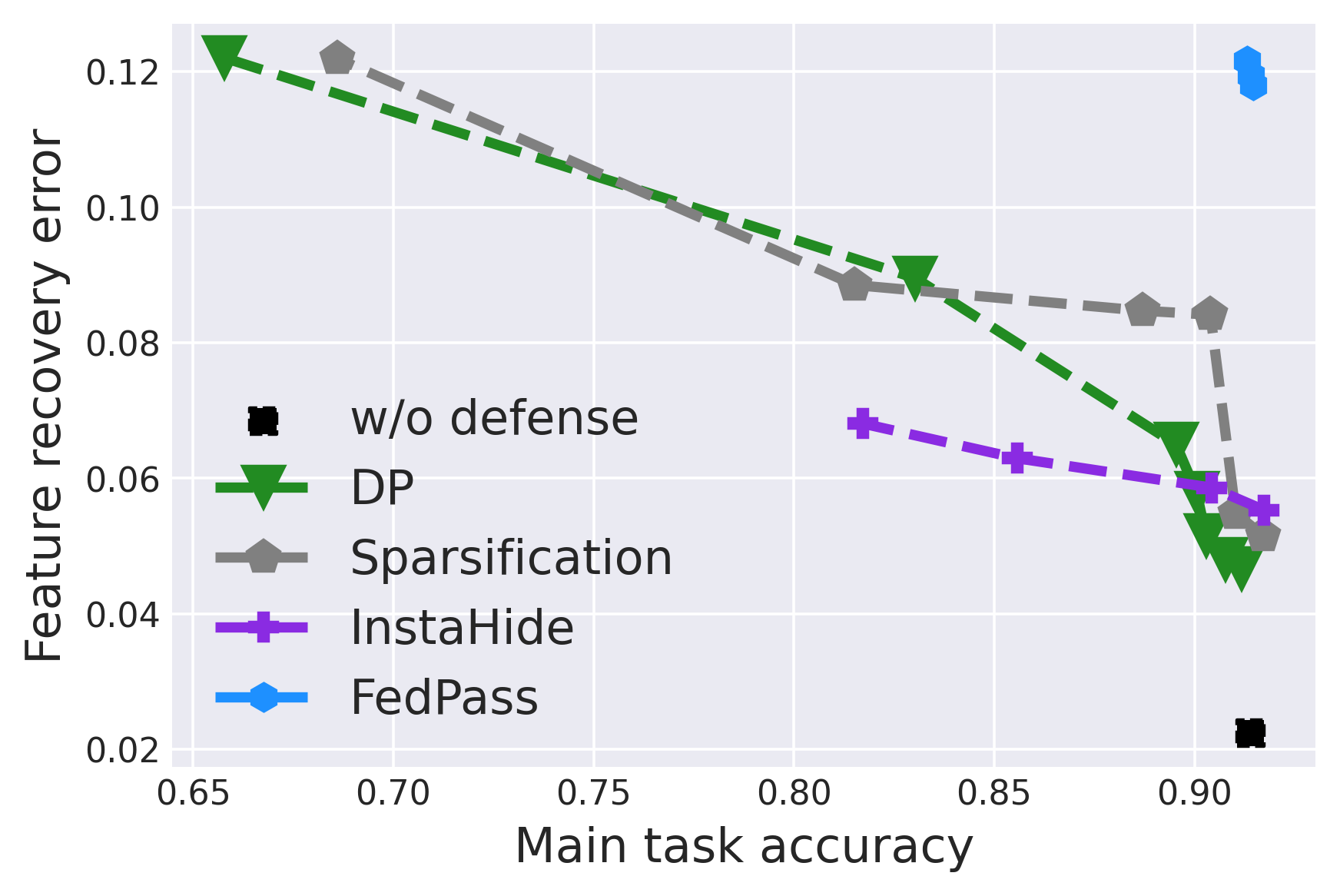}
         \subcaption{ResNet-CIFAR10}
		\end{subfigure}

      		\begin{subfigure}{0.3\textwidth}
  		 	\includegraphics[width=1\textwidth]{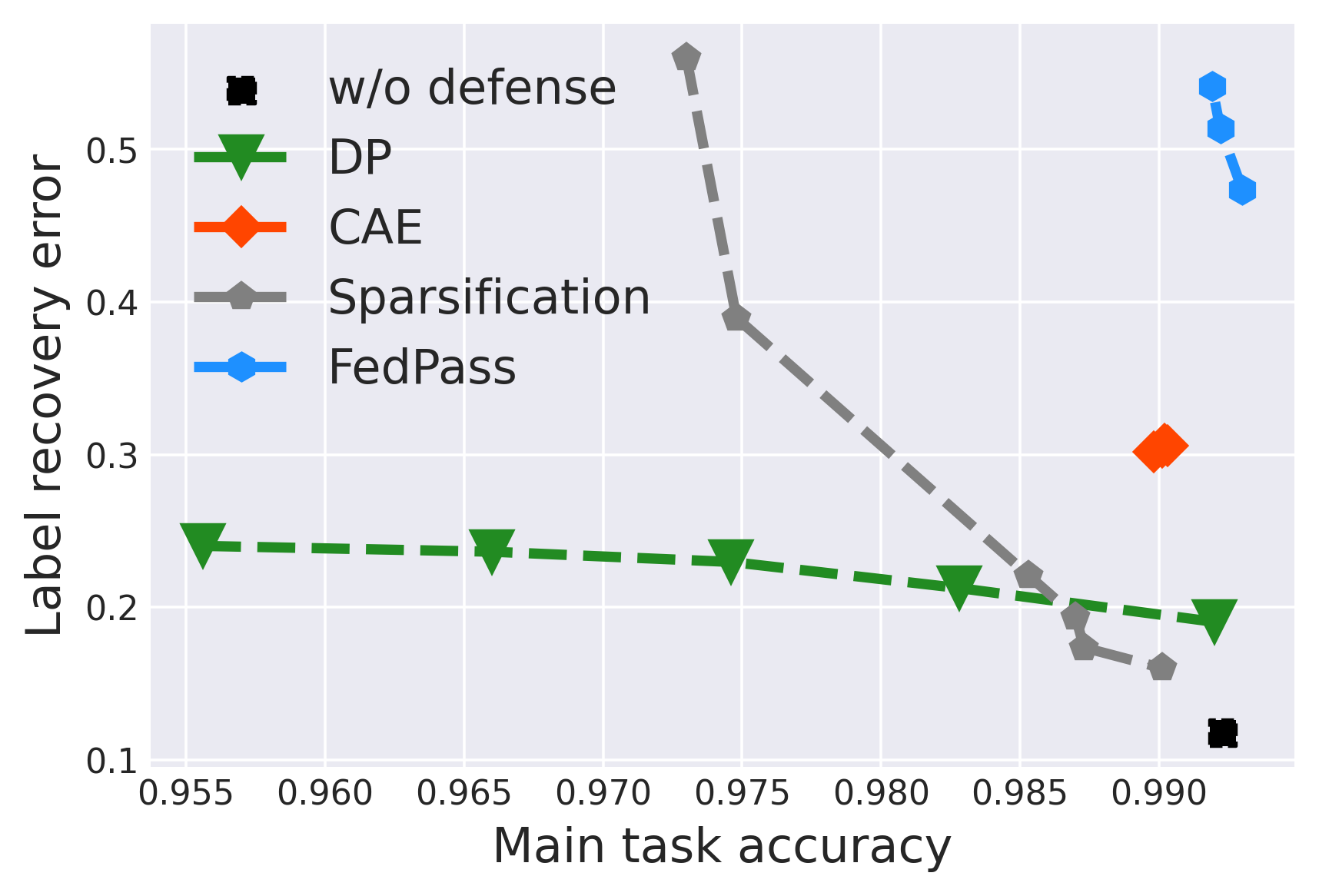}
      \subcaption{LeNet-MNIST}
    		\end{subfigure}
    	\begin{subfigure}{0.3\textwidth}
  		 	\includegraphics[width=1\textwidth]{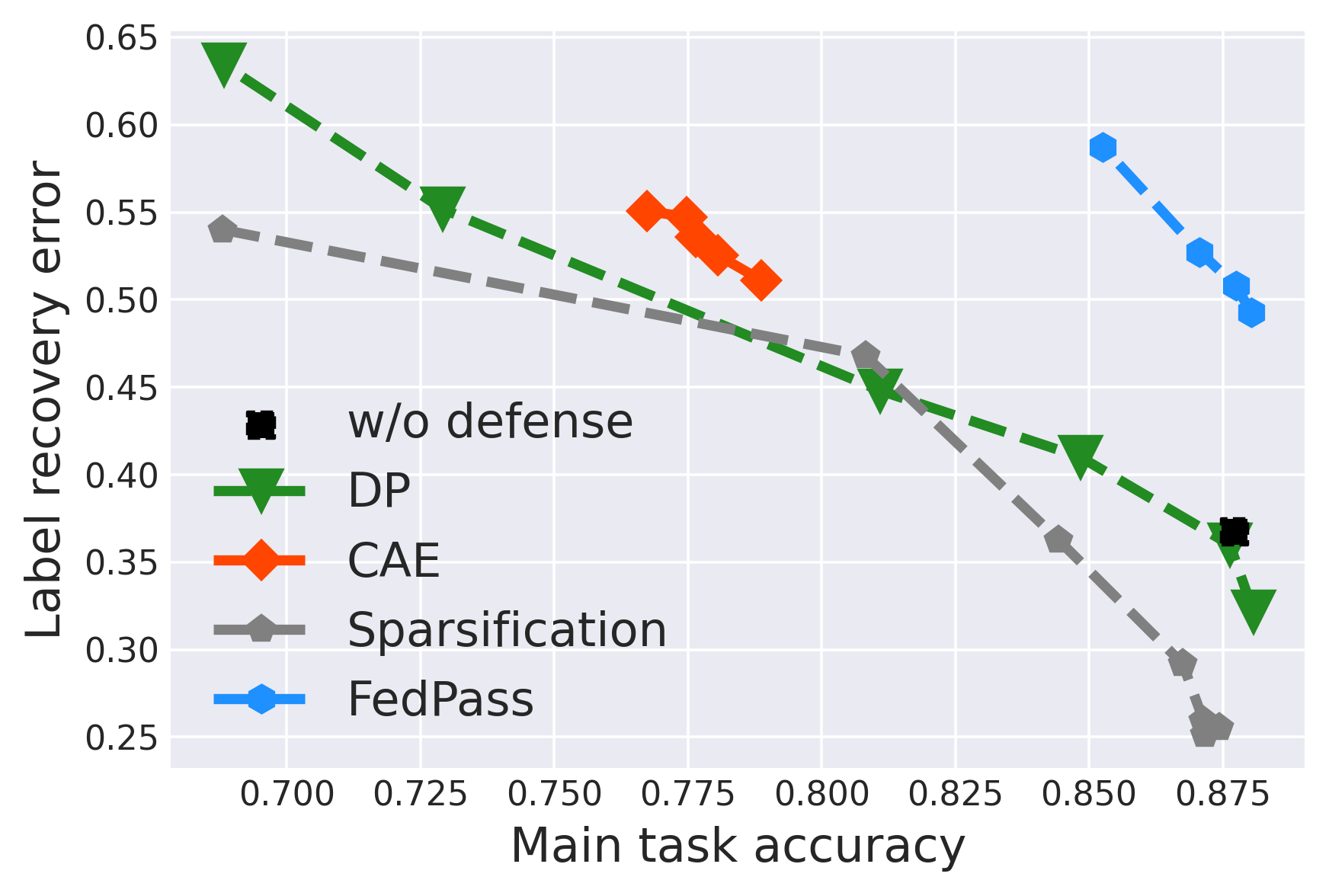}
            \subcaption{AlexNet-CIFAR10}
    		\end{subfigure}
   \begin{subfigure}{0.3\textwidth}
			\includegraphics[width=1\textwidth]{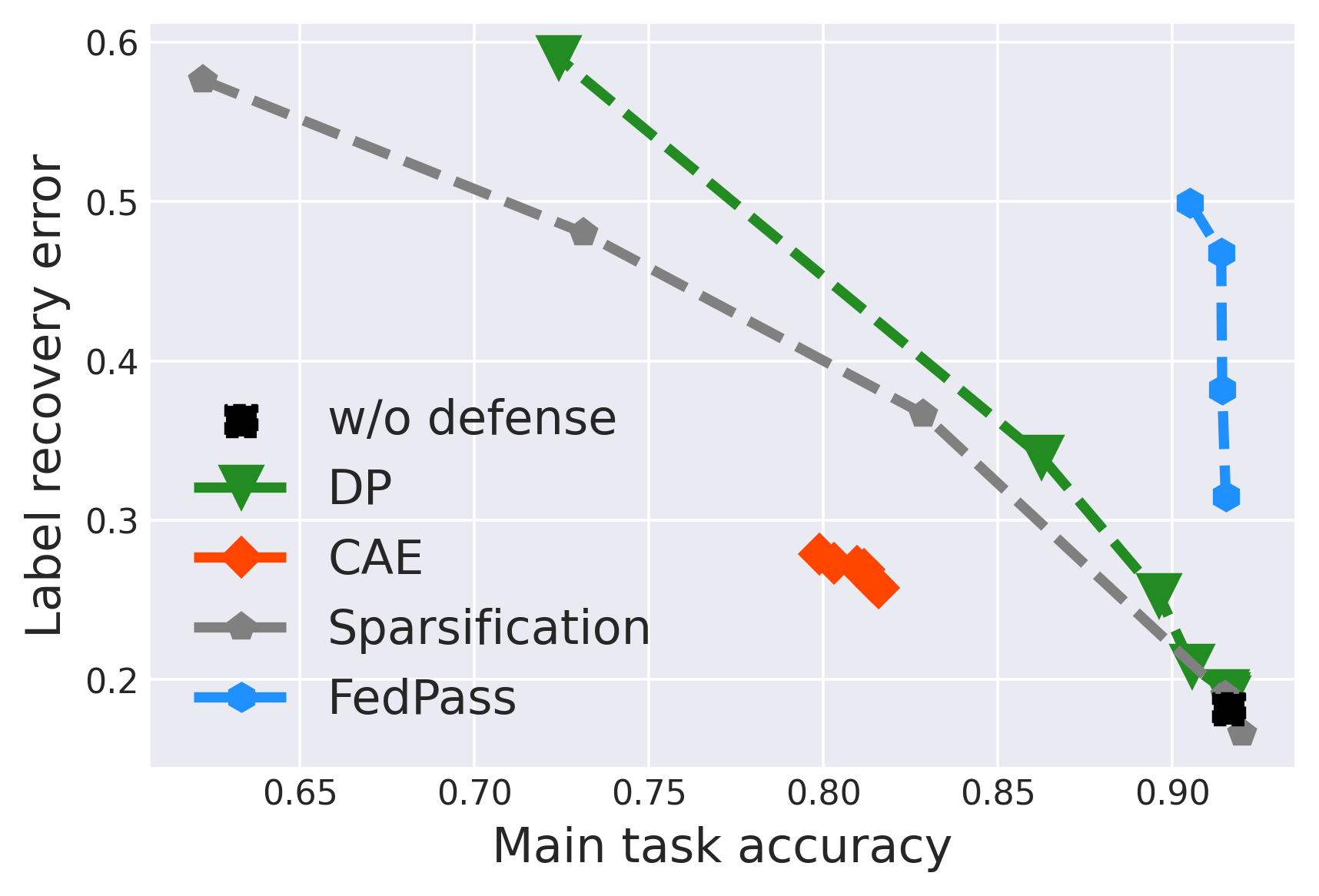}
         \subcaption{ResNet-CIFAR10}
		\end{subfigure}
\vspace{-0.8em}
 \caption{Comparison of different defense methods in terms of their trade-offs between main task accuracy and data (feature or label) recovery error against three attacks on LeNet-MNIST, AlexNet-CIFAR10 and ResNet-CIFAR10, respectively. \textbf{Model Inversion} (the first line) and \textbf{CAFE} (the second line) are feature reconstruction attacks, whereas \textbf{Passive Model Completion} (the third line) is a label inference attack. \textit{A better trade-off curve should be more toward the top-right corner of each figure}.}
	\label{fig:tradeoff_result}
 \vspace{-4pt}
\end{figure*}


\begin{table*}[!ht]
\center
\small
\setlength{\tabcolsep}{1.5mm}
\begin{tabular}{c|c||c|c|c|c|c|c}
\hline
\multicolumn{2}{c||}{\diagbox[dir=SW]{Defense}{Attack}} & w/o defense & CAE                      & Sparsification   & DP          & InstaHide     & FedPass       \\ \hline \hline
\multirow{3}{*}{CAFE}                    & LeNet                       & 0.033       & \textemdash              & 0.049$\pm$0.026 & 0.033$\pm$0.018 & 0.061$\pm$0.004 & \textbf{0.137$\pm$0.002} \\ 
                                         & AlexNet                   & 0.019       & \textemdash                             & 0.058$\pm$0.026 & 0.042$\pm$0.017 & 0.023$\pm$0.004 & \textbf{0.105$\pm$0.001} \\
                                         & ResNet                    & 0.021       & \textemdash                             & 0.067$\pm$0.014 & 0.057$\pm$0.014 & 0.053$\pm$0.002 & \textbf{0.109$\pm$0.001} \\ \hline
\multirow{3}{*}{MI}                      & LeNet                      & 0.033       & \textemdash                             & 0.060$\pm$0.020 & 0.049$\pm$0.010 & 0.046$\pm$0.005 & \textbf{0.087$\pm$0.001} \\
                                         & AlexNet                  & 0.043       & \textemdash                                 & 0.047$\pm$0.003 & 0.046$\pm$0.006 & 0.032$\pm$0.001 & \textbf{0.054$\pm$0.001} \\
                                         & ResNet                   & 0.046       & \textemdash                                 & 0.065$\pm$0.012 & 0.063$\pm$0.015 & 0.047$\pm$0.001 & \textbf{0.105$\pm$0.004} \\ \hline
\multirow{3}{*}{PMC}                      & LeNet                     & 0.117       & 0.302$\pm$0.002           & 0.277$\pm$0.140           & 0.216$\pm$0.015 & \textemdash           & \textbf{0.506$\pm$0.028} \\
                                         & AlexNet                 & 0.322       & 0.415 $\pm$0.008          & 0.283$\pm$0.065           & 0.358$\pm$0.051 & \textemdash            & \textbf{0.460$\pm$0.025} \\
                                         & ResNet                    & 0.166       & 0.217$\pm$0.004           & 0.268$\pm$0.088          & 0.237$\pm$0.087 & \textemdash            &  \textbf{0.379$\pm$0.065}  \\ \hline
\end{tabular}
\vspace{-0.8em}
\caption{The Calibrated averaged performance (CAP) for different defense mechanisms against CAFE, MI and PMC attacks.}
\label{tab:cap}
\end{table*}


  

\subsubsection{Privacy Attack methods}
We investigate the effectiveness of FedPass through three privacy attacks designed for VFL, namely, Passive Model Completion (PMC) attack \cite{fu2022label}, CAFE attack~\cite{jin2021cafe} and Model Inversion (MI) attack \cite{he2019model}. The first attack is a label inference attack, whereas the last two are feature reconstruction attacks (see details in Appendix A). 

\subsubsection{Baseline Defense Methods} 
We adopt four defense mechanisms as baselines to evaluate the effectiveness of FedPass in defending against feature reconstruction and label inference attacks. Each defense mechanism is controlled by a defense strength parameter to trade-off privacy leakage and model performance. For \textbf{Differential Privacy (DP)} \cite{abadi2016deep}, 
we experiment with Gaussian noise levels ranging from 5e-5 to 1.0. We add noise to gradients for defending against MC attack while add noise to forward embeddings for thwarting CAFE and MI attacks. 
For \textbf{Sparsification} \cite{fu2022label,lin2018deep}, we implement gradient sparsification \cite{fu2022label} and forward embedding sparsification \cite{lin2018deep} for defending against label inference attack and feature reconstruction attack, respectively. Sparsification level are chosen from 0.1\% to 50.0\%.  For \textbf{Confusional AutoEncoder} (CAE)~\cite{zou2022defending}, we follow the implementation of the original paper. That is, both the encoder and decoder of CAE have the architecture of 2 FC layers. Values of the hyperparameter that controls the confusion level are chosen from 0.0 to 2.0. For \textbf{InstaHide}~\cite{huang2020instahide}, we mix up 1 to 4 of images to trade-off privacy and utility. For \textbf{FedPass}, the range of the mean of Gaussian distribution $N$ is from 2 to 200, the variance is from 1 to 64. Passports are embedded in the last convolution layer of the passive party's model and first fully connected layer of the active party's model.


\subsubsection{Evaluation metrics} 
We use data (feature or label) recovery error and main task accuracy to evaluate defense mechanisms. We adopt the ratio of incorrectly labeled samples by a label inference attack to all labeled samples to measure the performance of that label inference attack. We adopt Mean Square Error (MSE) \cite{zhu2019dlg} between original images and images recovered by a feature reconstruction attack to measure the performance of that feature reconstruction attack. MSE is widely used to assess the quality of recovered images. A higher MSE value indicates a higher image recovery error. In addition, we leverage Calibrated Averaged Performance (CAP) \cite{fan2020rethinking} to quantify the trade-off between main task accuracy and data recovery error. CAP is defined as follows:
\begin{definition}[Calibrated Averaged Performance (CAP)] \label{def:cap}
For a given Privacy-Preserving Mechanism $g_s\in \calG$ ($s$ denotes the controlled parameter of $g$, e.g., the sparsification level, noise level and passport range) and attack mechanism $a \in \calA$, the Calibrated Averaged Performance is defined as:
\begin{equation}
    \text{CAP}(g_s, a) = \frac{1}{m}\sum_{s=s_1}^{s_m} Acc(g_s, x) * Rerr(x, \hat{x}_s),
\end{equation}
where $Acc(\cdot)$ denotes the main task accuracy and $Rerr(\cdot)$ denotes the recovery error between original data $x$ and estimated data $\hat{x}_s$ via attack $a$.
\end{definition}



\subsection{Experiment Results} \label{sec:exp}

\subsubsection{Defending against the Feature Reconstruction Attack}
 
Figure \ref{fig:tradeoff_result} (a)-(f) compare the trade-offs between feature recovery error (y-axis) and main task accuracy (x-axis) of FedPass with those of baselines against MI and CAFE attacks on three models. We observe: \romannumeral1) DP and Sparsification can generally achieve either a high main task performance or a large feature recovery error (low privacy leakage), but not both. For example, DP and Sparsification can achieve a main task performance as high as $\geq 0.90$ while obtaining a feature recovery error as low as $\leq 0.06$ on ResNet-CIFAR10. At the other extreme, DP and Sparsification can achieve $\geq 0.11$ feature recovery error but obtain $\leq 0.70$ main task performance on ResNet-CIFAR10. \romannumeral2) InstaHide generally can not thwart MI and CAFE attacks. Even mixing up with more data, InstaHide still leads to a relatively small feature recovery error while its main task performance degrades significantly. \romannumeral3) The trade-off curves of FedPass reside near the top-right corner under both attacks on all models, indicating that FedPass achieves the best performance on preserving feature privacy while maintaining the model performance. For example, FedPass achieves $\geq 0.91$ main task accuracy and $\geq 0.12$  feature recovery error under MI and CAFE attacks on ResNet-CIFAR10. Table \ref{tab:cap} also demonstrates that FedPass has the best trade-off between privacy and performance under MI and CAFE attacks.





Figure \ref{fig:vis-whitebox} showcases that, when protected by FedPass ($r8$), reconstructed images under the CAFE attack on all three datasets are essentially random noise, manifesting that FedPass can thwart the CAFE attack effectively. At the same time, the model performance of FedPass is almost lossless compared to that of the original model (w/o defense) on every dataset (see Figure \ref{fig:tradeoff_result}). This superior trade-off of FedPass is in sharp contrast to existing methods, among which DP and Sparsification with high protection levels ($r5$ and $r7$) lead to significantly deteriorated model performance.

\begin{figure}[!ht]
\vspace{-3pt}
	\centering
 \captionsetup[subfigure]{labelformat=empty}
	{
   $r1$
		\begin{subfigure}{0.141\textwidth}
\includegraphics[width=1\textwidth]{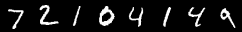}
		\end{subfigure}
	}
    	{
    		\begin{subfigure}{0.141\textwidth}
  		 	\includegraphics[width=1\textwidth]{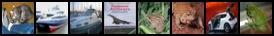}
    		\end{subfigure}
    	}
    	{
    		\begin{subfigure}{0.141\textwidth}
  		 	\includegraphics[width=1\textwidth]{imgs/exp/reconstruct_whitebox_new/cifar_8_oriimg.png}
    		\end{subfigure}
    	}\vspace{-1pt}

     {
            $r2$
		\begin{subfigure}{0.141\textwidth}
\includegraphics[width=1\textwidth]{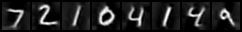}
		\end{subfigure}
	}
      {
		\begin{subfigure}{0.141\textwidth}
\includegraphics[width=1\textwidth]{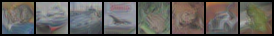}
		\end{subfigure}
	}
      {
		\begin{subfigure}{0.141\textwidth}
\includegraphics[width=1\textwidth]{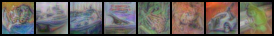}
		\end{subfigure}
	}\vspace{-1pt}

      {
             $r3$
		\begin{subfigure}{0.141\textwidth}
\includegraphics[width=1\textwidth]{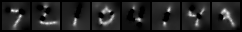}
		\end{subfigure}
	}
      {
		\begin{subfigure}{0.141\textwidth}
\includegraphics[width=1\textwidth]{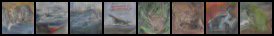}
		\end{subfigure}
	}
      {
		\begin{subfigure}{0.141\textwidth}
\includegraphics[width=1\textwidth]{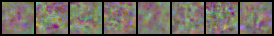}
		\end{subfigure}
	}\vspace{-1pt}

      {
            $r4$
		\begin{subfigure}{0.141\textwidth}
\includegraphics[width=1\textwidth]{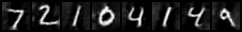}
		\end{subfigure}
	}
      {
		\begin{subfigure}{0.141\textwidth}
\includegraphics[width=1\textwidth]{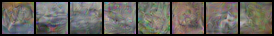}
		\end{subfigure}
	}
      {
		\begin{subfigure}{0.141\textwidth}
\includegraphics[width=1\textwidth]{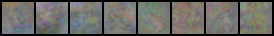}
		\end{subfigure}
	}\vspace{-1pt}

      {   $r5$
		\begin{subfigure}{0.141\textwidth}
\includegraphics[width=1\textwidth]{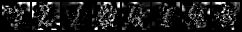}
		\end{subfigure}
	}
      {
		\begin{subfigure}{0.141\textwidth}
\includegraphics[width=1\textwidth]{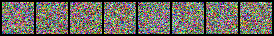}
		\end{subfigure}
	}
      {
		\begin{subfigure}{0.141\textwidth}
\includegraphics[width=1\textwidth]{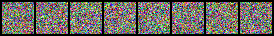}
		\end{subfigure}
	}\vspace{-1pt}

      {
         $r6$
		\begin{subfigure}{0.141\textwidth}
\includegraphics[width=1\textwidth]{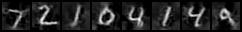}
		\end{subfigure}
	}
      {
		\begin{subfigure}{0.141\textwidth}
\includegraphics[width=1\textwidth]{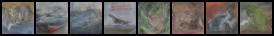}
		\end{subfigure}
	}
      {
		\begin{subfigure}{0.141\textwidth}
\includegraphics[width=1\textwidth]{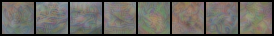}
		\end{subfigure}
	}\vspace{-1pt}

      {
        $r7$
		\begin{subfigure}{0.141\textwidth}
\includegraphics[width=1\textwidth]{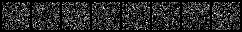}
		\end{subfigure}
	}
      {
		\begin{subfigure}{0.141\textwidth}
\includegraphics[width=1\textwidth]{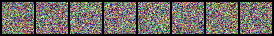}
		\end{subfigure}
	}
      {
		\begin{subfigure}{0.141\textwidth}
\includegraphics[width=1\textwidth]{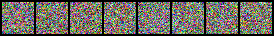}
		\end{subfigure}
	}\vspace{-1.2pt}

       {
        $r8$
\begin{subfigure}{0.141\textwidth}
\includegraphics[width=1\textwidth]{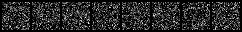}
		\end{subfigure}
	}
      {
		\begin{subfigure}{0.141\textwidth}
\includegraphics[width=1\textwidth]{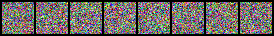}
		\end{subfigure}
	}
      {
		\begin{subfigure}{0.141\textwidth}
\includegraphics[width=1\textwidth]{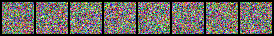}
		\end{subfigure}
	}

       {
\begin{subfigure}{0.16\textwidth}
\subcaption*{\tiny LeNet-MNIST}
		\end{subfigure}
	}
      {
		\begin{subfigure}{0.146\textwidth}
\subcaption*{\tiny AlexNet-CIFAR10}
		\end{subfigure}
	}
      {
		\begin{subfigure}{0.146\textwidth}
\subcaption*{\tiny ResNet-CIFAR10}
		\end{subfigure}
	}
        \vspace{-1.8em}
	\caption{
 Original images and images reconstructed by CAFE attack for different defense mechanisms on LeNet-MNIST, AlexNet-CIFAR10 and ResNet-CIFAR10, respectively. From top to bottom, a row represents original image ($r1$), no defense ($r2$), InstaHide ($r3$), DP with noise level $0.2$ ($r4$) and $2$ ($r5$), Sparsification with sparsification level $0.5$ ($r6$) and $0.05$ ($r7$), and FedPass ($r8$).} 
	\label{fig:vis-whitebox}
	\vspace{-6pt}
\end{figure}

\subsubsection{Defending against the Label Inference Attack}


Figure \ref{fig:tradeoff_result} (g)-(i) compare the trade-offs between label recovery error (y-axis) and main task accuracy (x-axis) of FedPass with those of baselines in the face of the PMC attack on three models. It is observed DP, Sparsification, and CAE fail to achieve the goal of obtaining a low level of privacy leakage while maintaining model performance, whereas FedPass is more toward the top-right corner, indicating that FedPass has a better trade-off between privacy and performance. Table \ref{tab:cap} reinforce the observation that FedPass achieves the best trade-off between privacy and performance under PMC attack.



\subsubsection{Training and Inference Time}
Table \ref{tab:time} investigates the training time (for one epoch) and inference time for FedPass and baseline defense methods. It shows that the FedPass is as efficient as the VFL w/o defense for both training and inference procedures (the training time on MNIST for each epoch is 7.03s and inference time is 1.48s) because embedding passport only introduces a few model parameters to train. It is worth noting that the training time of InstaHide is almost twice that of other methods because InstaHide involves mixing-up multiple feature vectors or labels, which is time-consuming.
\begin{table}[htbp]
\footnotesize
\center
\begin{tabular}{@{}lrlllll@{}}
\toprule
\multirow{2}{*}{\begin{tabular}[c]{@{}l@{}}Defense\\ Method\end{tabular}} &
  \multicolumn{2}{l}{\scriptsize LeNet-MNIST} &
  \multicolumn{2}{l}{\scriptsize AlexNet-Cifar10}  & \multicolumn{2}{l}{\scriptsize ResNet-Cifar10}\\ \cmidrule(l){2-7} 
 & Train & Infer & Train & Infer & Train & Infer \\ \hline \hline
w/o defense & 7.03  & 1.48 & 22.37 & 2.20 & 22.64 & 2.18 \\
\hline
CAE         & 7.30  & 1.48 & 22.71 & 2.27 & 23.02 & 2.21 \\
Sparsification & 6.93  & 1.45 & 22.39 & 2.12 & 22.61 & 2.21 \\
DP       & 7.01  & 1.49 & 22.24 & 2.23 & 22.63 & 2.16 \\
InstaHide   & 21.76 & 1.50 & 37.07 & 2.19 & 46.26 & 2.18 \\
\hline
\\[-1em]
FedPass (ours)    & 7.05  & 1.46 & 22.58 & 2.13 & 22.61 & 2.16 \\ \bottomrule
\end{tabular}
\caption{Comparison of training time (for one epoch) and inference time among different defense mechanisms.}
\vspace{-1em}
\label{tab:time}
\end{table}

\section{Conclusion}
This paper proposes a novel privacy-preserving vertical federated deep learning framework called FedPass, which leverages adaptive obfuscation to protect the label and data simultaneously. Specifically, the proposed adaptive obfuscation is implemented by embedding private passports in the passive and active models to 
adapt the deep learning model such that the model performance is preserved. The extensive experiments on multiple datasets and theoretical analysis demonstrate that the FedPass can achieve significant improvements over the other protected methods in terms of model performance and privacy-preserving ability.
\vspace{20pt}


\appendix

\setcounter{equation}{0}
\setcounter{theorem}{0}
\setcounter{prop}{0}
\setcounter{definition}{0}

\noindent \textbf{\huge Appendix}
\section{Experimental Setting}

This section provides detailed information on our experimental settings. Table \ref{table:models-app} summarizes datasets and DNN models evaluated in all experiments and Table \ref{tab:train-params} summarizes the hyper-parameters used for training  models.



\subsection{Dataset \& Model Architectures}
 We consider classification tasks by LeNet \cite{lecun1998gradient} on MNIST \cite{lecun2010mnist}, AlexNet \cite{NIPS2012_c399862d} on CIFAR10 \cite{krizhevsky2014cifar}, ResNet \cite{he2016deep} on CIFAR10 dataset and LeNet on ModelNet \cite{wu20153d,liu2022cross}. The MNIST database of 10-class handwritten digits has a training set of 60,000 examples, and a test set of 10,000 examples. The CIFAR-10 dataset consists of 60000 $32\times 32$ colour images in 10 classes, with 6000 images per class. ModelNet is a widely-used 3D shape classification and shape retrieval benchmark, which currently contains 127,915 3D CAD models from 662 categories. We created 12 2D multi-view images per 3D mesh model by placing 12 virtual cameras evenly distributed around the centroid and partitioned the images into multiple (2 to 6) parties by their angles, which contains 6366 images to train and 1600 images to test. Moreover, we add the batch normalization layer of the LeNet and AlexNet for after the last convolution layer for existing defense methods.
\begin{table}[!htbp]
	\centering
	\footnotesize
	\begin{tabular}{c||c|c|c|c}
	        \hline
             \\[-1em]
		\shortstack{Dataset \\ Name} & \shortstack{Model \\ Name}  & \shortstack{Model of \\ Passive Party } & \shortstack{Model of \\ Active Party}   & \# P \\
         \hline
         \hline
          \\[-1em]
            MNIST &  LeNet & 2 Conv  &  3 FC  & 2 \\
		\hline
          \\[-1em]
		CIFAR10 & AlexNet & 5 Conv  &  1 FC & 2  \\
		\hline
          \\[-1em]
		CIFAR10 & ResNet18 & 17 Conv & 1 FC & 2\\
     \hline
     \\[-1em]
     	ModelNet & LeNet & 2 Conv  &  3 FC  & 7 \\
		\hline
         
	\end{tabular}
 \vspace{-0.6em}
	\caption{Models for evaluation. \# P denotes the number of parties. FC: fully-connected layer. Conv: convolution layer. 
	}
\label{table:models-app}
\vspace{-3pt}
\end{table}

\begin{table*}[!htbp]  \renewcommand\arraystretch{1.2}
\vspace{-3pt}
	\resizebox{1\textwidth}{!}{
		\begin{tabular}{l|cccc}
			\hline
			Hyper-parameter & LeNet-MNIST& AlexNet-CIFAR10 & ResNet-CIFAR10& LeNet-ModelNet \\ \hline
			Optimization method & SGD & SGD& SGD& SGD\\
			Learning rate & 1e-2 & 1e-2& 1e-2&1e-2\\
			Weight decay & 4e-5 & 4e-5& 4e-5&4e-5\\
			Batch size & 64 & 64& 64&64\\
			Iterations & 50 & 100& 100&100 \\
			The range of $N$ (passport mean) & [1, 50] & [1, 200] & [1, 100] & [1, 50] \\
			The range of $\sigma^2$ (passport variance) & [1, 9] & [1, 64] & [1, 25] & [1,  9] \\
			\hline
	\end{tabular}}
	
	\caption{Hyper-parameters used for training in FedPass.}
	\label{tab:train-params}
\end{table*}

\subsection{Federated Learning Settings}
We simulate a VFL scenario by splitting a neural network into a bottom model and a top model and assigning the bottom model to each passive party and the top model to the active party. For the 2-party scenario, the passive party has features when the active party owns the corresponding labels. For the 7-party scenario, i.e., the ModelNet, each passive party has two directions of 3D shapes when the active party owns the corresponding label following \cite{liu2022cross}. Table \ref{table:models-app} summarizes our VFL scenarios. Also, the VFL framework follows algorithm 1 of \cite{liu2022vertical}.


\subsection{Privacy Attack methods}
We investigate the effectiveness of FedPass on three attacks designed for VFL
\begin{itemize}
    \item Passive Model Completion (PMC) attack \cite{fu2022label} is the label inference attack. For each dataset, the attacker leverages some auxiliary labeled data (40 for MNIST and CIFAR10, 366 for ModelNet) to train the attack model and then the attacker predicts labels of the test data using the trained attack model. 
    \item CAFE attack \cite{jin2021cafe} is a feature reconstruction attack. We assume the attacker knows the forward embedding $H$ and passive party's model $G_{\theta}$, and we follow the white-box model inversion step (i.e., step 2) of CAFE to recover private features owned by a passive party.
    \item Model Inversion (MI) attack is a feature reconstruction attack. We follow the black-box setting of \cite{he2019model}, in which the attacker does not know the passive model. For each dataset, attacker leverages some labeled data (10000 for MNIST and CIFAR10, 1000 for ModelNet) to train a shallow model approximating the passive model, and then they use the forward embedding and the shallow model to infer the private features of a passive party inversely.  
\end{itemize}
For CAFE and MI attacks, we add the total variation loss~\cite{jin2021cafe} into two feature recovery attacks with regularization parameter $\lambda=0.1$.
\subsection{Implementation details of FedPass}
For \textbf{FedPass}, Passports are embedded in the last convolution layer of the passive party's model and the first fully connected layer of the active party's model. Table \ref{tab:train-params} summarizes the hyper-parameters for training FedPass.

\subsection{Notations}
\begin{table}[!htbp] 
  \renewcommand{\arraystretch}{1.05}
  \centering
  \setlength{\belowcaptionskip}{15pt}
  \label{table: notation}
    \begin{tabular}{c|p{5.5cm}}
    \toprule
    Notation & Meaning\cr
    \midrule\
    $F_\omega, G_{\theta_k}$ & Active model and $j_{th}$ passive model\cr \hline
     $W_p, W_a, W_{att}$ & The 2D matrix of the active model, passive model and attack model for the regression task \cr \hline
     $s^a, s^p$ & The passport of the active model and passive model \cr \hline
    $s^a_\gamma, s^a_\beta,s^p_\gamma, s^p_\beta$ & The passport of the active model and passive model w.r.t $\gamma$ and $\beta$ respectively \cr \hline
    $\gamma, \beta$ & The scaling and bias according to Eq. \eqref{eq:pst1-app}  \cr \hline
    $g_W$ & Adaptive obfuscations w.r.t model parameters $W$ \cr \hline
    $H_k$ & Forward embedding passive party $k$ transfers to the active party \cr \hline
    $\Tilde{\ell}$ & the loss of main task \cr \hline
    $\nabla_{H_k}\Tilde{l}$ & Backward gradients the active party transfers to the passive party $k$ \cr \hline
    $\calD_k= \{x_{k,i}\}_{i=1}^{n_i}$ & $n_i$ private features of passive party $k$ \cr \hline
     $y$ & Private labels of active party \cr \hline
    $K$ & The number of party\cr \hline
    $N$ & The range of Gaussian mean of passports sample \cr \hline
   $\sigma^2$ & The variance of passports sample \cr \hline
   $n_a$ & The number of auxiliary dataset to do PMC attack\cr \hline
    $\Tilde{\ell}_a$ &  Training error on the auxiliary dataset for attackers \cr \hline
        $\Tilde{\ell}_t$ &  Test error on the test dataset for attackers \cr \hline
      $n_a$ & The number of auxiliary dataset to do PMC attack\cr \hline
   $T$ & Number of optimization steps for VFL \cr \hline
      $\eta$ & learning rate \cr \hline
 $\|\cdot \| $ & $\ell_2$ norm \cr 
    \bottomrule
    \end{tabular}
    \caption{Table of Notations}
\end{table}


	
			


\newpage
\section{More Experiments}
This section reports experiments on the multi-party scenario of VFL and ablation study on two critical passport parameters.
\subsection{FedPass on Multi-party Scenario}
We conduct experiments with 7 parties on ModelNet with six passive parties and one active parties. Figure \ref{fig:tradeoff_result_modelnet} compares the trade-offs between recovery error (y-axis) and main task accuracy (x-axis) of FedPass with those of baselines against MI, CAFE, and PMC attacks on three models. We observe: 
\begin{itemize}
    \item DP and Sparsification can generally achieve either a high main task performance or a large recovery error (low privacy leakage), but not both. For example, DP and Sparsification can achieve a main task performance as high as $\geq 0.775$ while obtaining the label recovery error as low as $0.56$ respectively. At the other extreme, DP and Sparsification can achieve $\geq 0.64$ feature recovery error but obtain $\leq 0.70$ main task performance.
    \item InstaHide generally can not thwart MI and CAFE attacks. Even mixing up with more data, InstaHide still leads to a relatively small feature recovery error while its main task performance degrades significantly. CAE also leads the drop of model performance.
    \item The trade-off curves of FedPass reside near the top-right corner under both attacks on all models, indicating that FedPass achieves the best performance on preserving privacy while maintaining the model performance. For example, FedPass achieves $\geq 0.775$ main task accuracy, $\geq 0.054$ feature recovery error under MI attack and $\geq 0.66$ label recovery error under MI attack. 
\end{itemize}

\begin{figure*}[ht]
\vspace{-3pt}
	\centering
   		\begin{subfigure}{0.3\textwidth}
  		 	\includegraphics[width=1\textwidth]{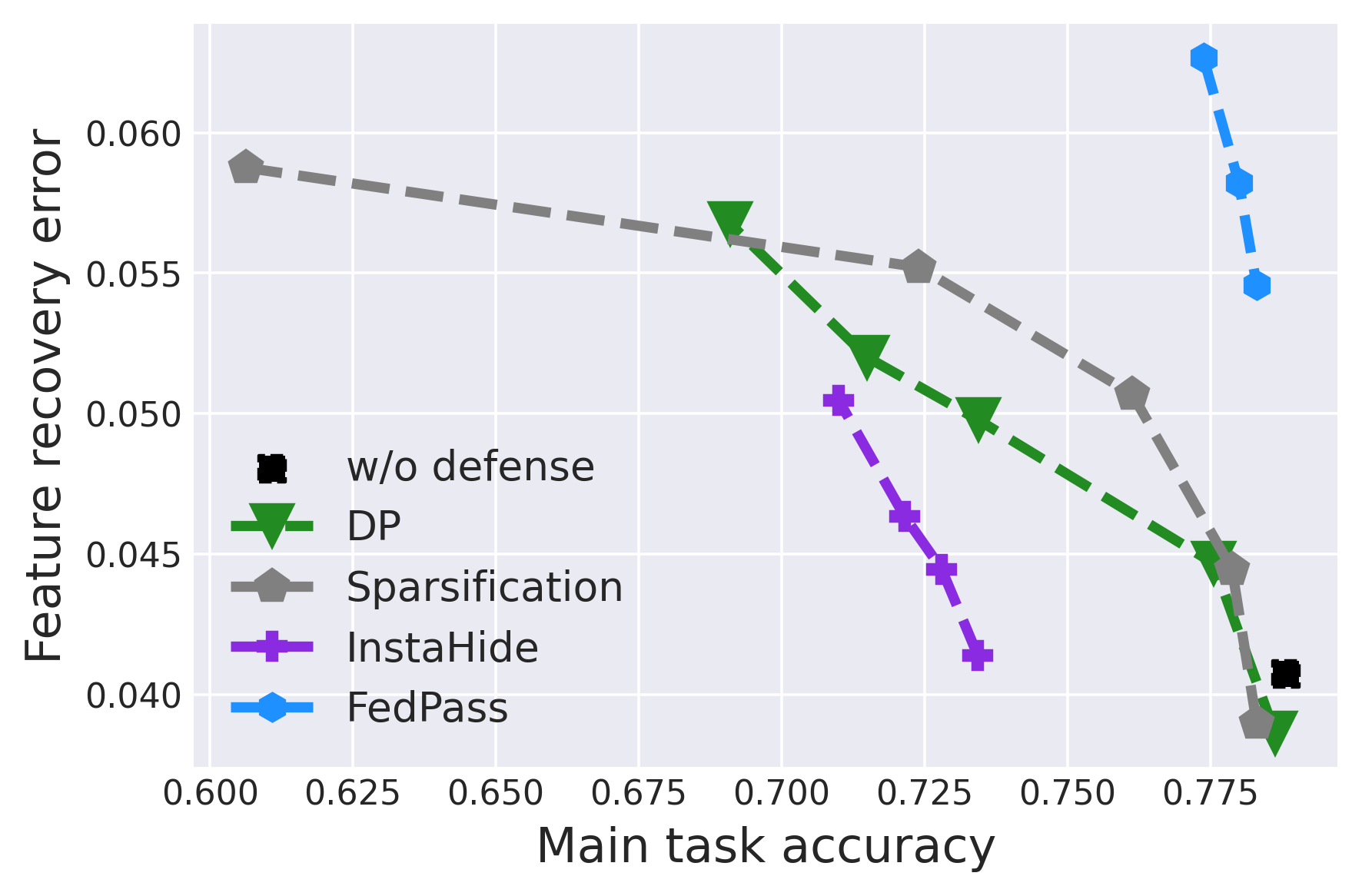}
      \subcaption{}
    		\end{subfigure}
    	\begin{subfigure}{0.3\textwidth}
  		 	\includegraphics[width=1\textwidth]{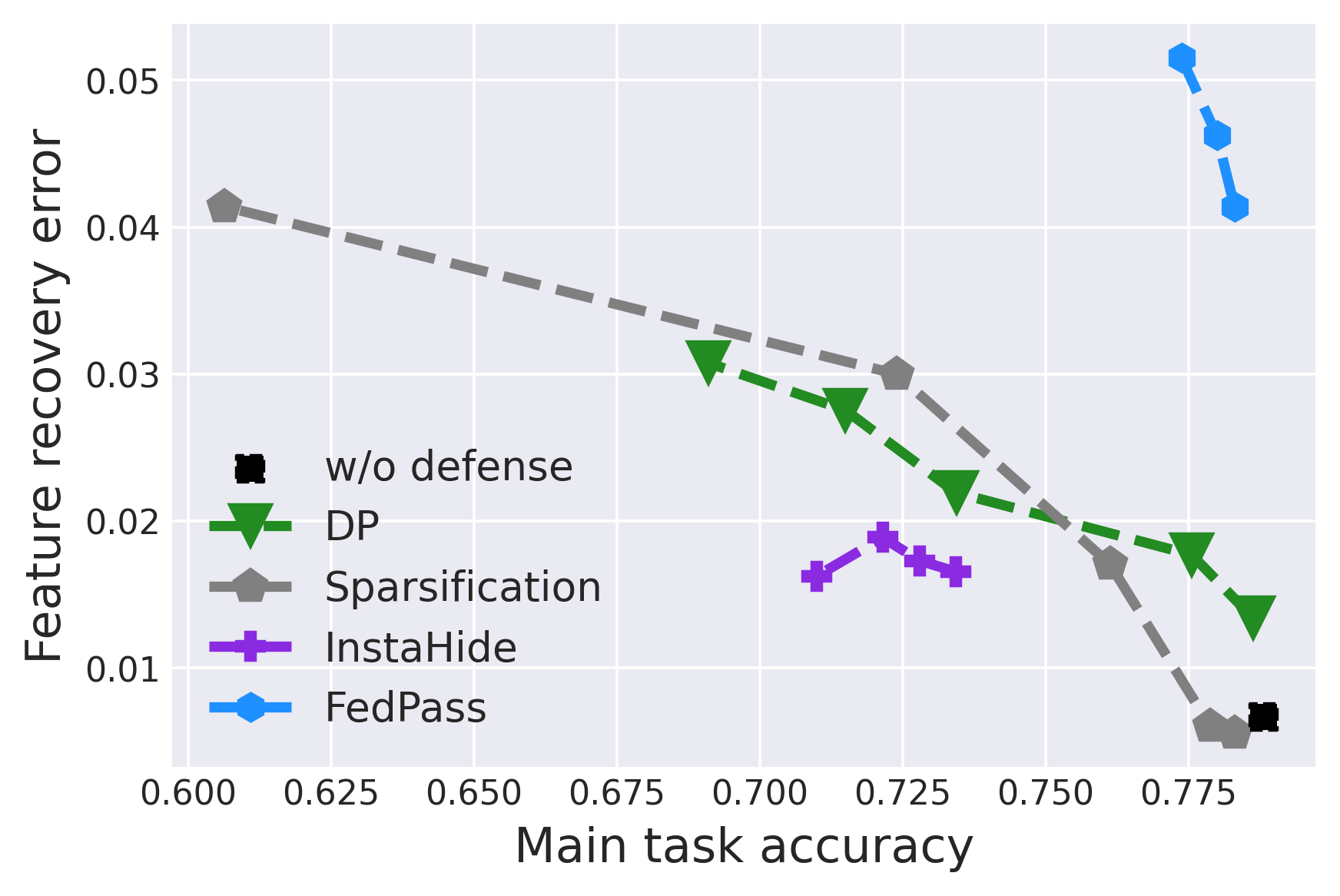}
            \subcaption{}
    		\end{subfigure}
   \begin{subfigure}{0.3\textwidth}
			\includegraphics[width=1\textwidth]{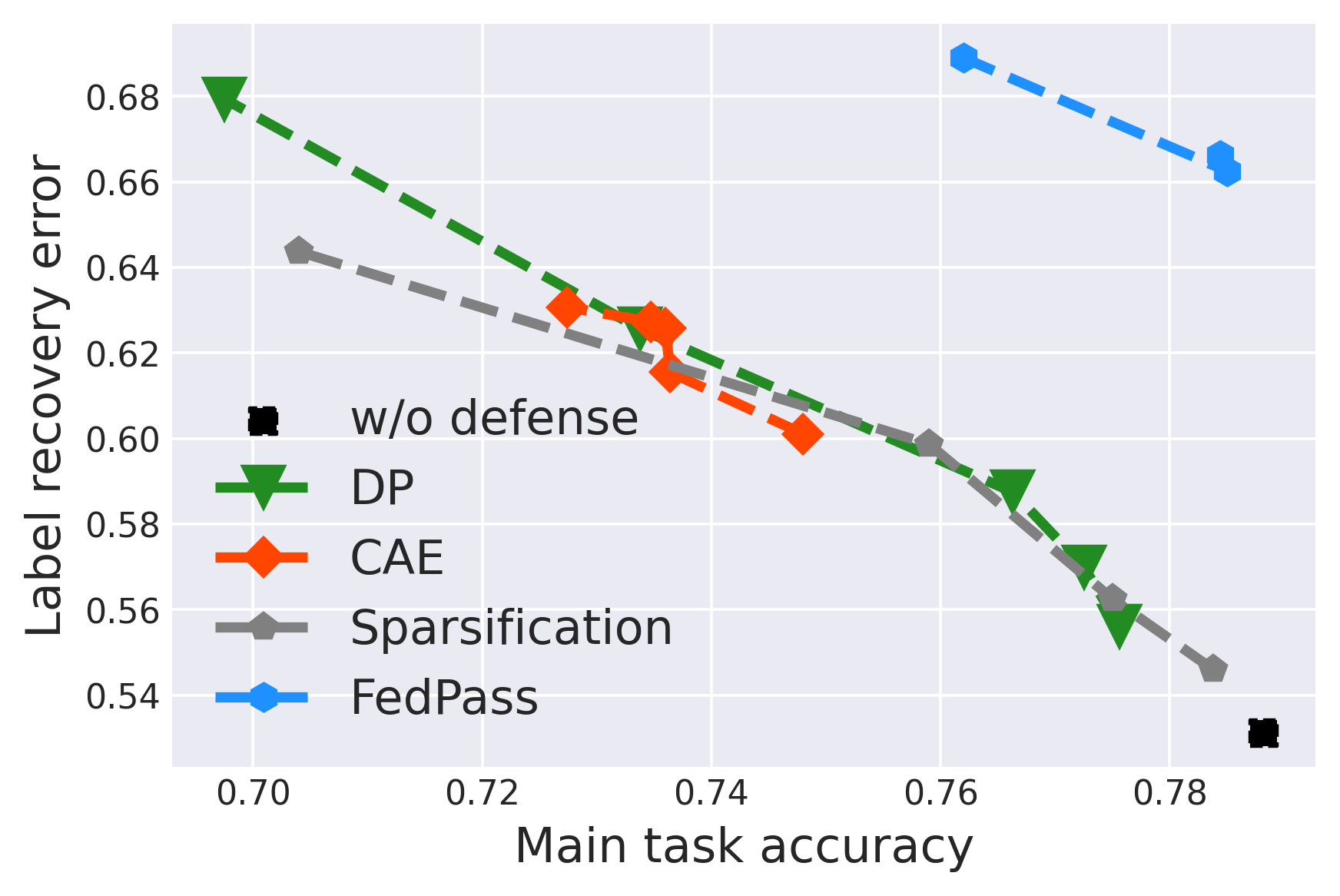}
         \subcaption{}
		\end{subfigure}
\vspace{-0.8em}
 \caption{Comparison of different defense methods in terms of their trade-offs between main task accuracy and data (feature or label) recovery error against three attacks on LeNet-ModelNet, which involves \textbf{7 parties}. \textbf{Model Inversion} (Figure (a)) and \textbf{CAFE} (Figure (b)) are feature reconstruction attacks, whereas \textbf{Passive Model Completion} (Figure (c)) is a label inference attack. \textit{A better trade-off curve should be more toward the top-right corner of each figure}.}
	\label{fig:tradeoff_result_modelnet}
 \vspace{-4pt}
\end{figure*}
\subsection{Ablation Study} \label{sec:aba-tradeoff}
As illustrated in main text, when the passports are embedded in a convolution layer or linear layer with $c$ channels\footnote{For the convolution layer, the passport $s\in \RR^{c\times h_1 \times h_2}$, where $c$ is channel number, $h_1$ and $h_2$ are height and width; for the linear layer, $s\in \RR^{ c\times h_1}$, where $c$ is channel number, $h_1$ is height.}, 
for each channel $j\in[c]$, 
the passport $s{(j)}$ (the $j_{th}$ element of vector $s$) is randomly generated as follows:
\begin{equation}\label{eq:sample-pst-app}
    s{(j)} \sim \calN(\mu_j, \sigma^2), \quad
    \mu_j \in \calU(-N, 0),
\end{equation}
where all $\mu_j, j=1,\cdots, c$ are different from each other, $\sigma^2$ is the variance of Gaussian distribution and $N$ is \textit{the range of passport mean}.  $\sigma$ and $N$ are two crucial parameters determining the privacy-utility trade-off of FedPass. This section analyzes how these two parameters influence the data recovery error and main task accuracy for FedPass. 

\begin{table}[htbp]
\centering
\resizebox{0.3\textwidth}{!}{\begin{tabular}{@{}c|c|c@{}}
\cmidrule(r){1-3}
         $N$ & \shortstack{Main task \\accuracy} & \shortstack{Feature recovery \\error}   \\ \cmidrule(r){1-3}
1   & 0.8798        & 0.0589                  \\
2   & 0.8771        & 0.0818                   \\
5   & 0.8758        & 0.1203           \\
10  & 0.8764        & 0.1204                   \\
\textbf{50}  & 0.8761        & 0.1205                   \\
100 & 0.8759        & 0.1205                   \\
200 & 0.8739        & 0.1205                   \\
\cmidrule(r){1-3}
\end{tabular}}
\caption{Influence of range of passports mean $N$ on main task accuracy and feature recovery error for FedPass against the CAFE attack.}
\label{tab:gaussian-mean}
\end{table}
\subsubsection{Influence of the range of Passport Mean $N$}
For the range of passport mean $N$, we consider the 2-party VFL scenario in AlexNet-CIFAR10 under feature recovery attack, i.e., CAFE \cite{jin2021cafe}, and MI \cite{he2019model} attacks. Fixed the $\sigma=1$, Table \ref{tab:gaussian-mean} provides the trade-off between main task accuracy and feature recovery error with different passport mean, which illustrates that when the range of passport mean increases, the privacy is preserved at the expense of the minor degradation of model performance less than 1\%. This phenomenon is consistent with our analysis of Theorem \ref{thm:thm1-app}. It is suggested that the $N$ is chosen as 100 when the model performance is less than 0.5\%.

\subsubsection{Influence of Passport Variance $\sigma^2$}
For the passport variance $\sigma^2$, we consider the 2-party VFL scenario in AlexNet-CIFAR10 under label recovery attack \cite{fu2022label}, i.e., PMC attack \cite{fu2022label}. Fixed the range of passport mean $N =100$, Table \ref{tab:gaussian-var} shows that the main task accuracy decreases and label recovery error increases when the passport variance increases, which is also verified in Theorem \ref{thm2-app}. It is suggested that the $\sigma^2$ is chosen as 5 when the model performance is less than 0.5\%.

\begin{table}[htbp]
\centering
\resizebox{0.3\textwidth}{!}{\begin{tabular}{c|c|c}
\toprule
$\sigma^2$ & \shortstack{Main task\\ accuracy} & \shortstack{Label recovery \\ error} \\ \midrule
0        & 0.877         & 0.368              \\
2        & 0.873         & 0.510               \\
\textbf{5}        & 0.870         & 0.543              \\
10       & 0.864         & 0.568              \\
50       & 0.858         & 0.577              \\
70       & 0.844         & 0.617              \\
100      & 0.791         & 0.751              \\ \bottomrule
\end{tabular}}
\caption{Influence of passports variance $\sigma^2$ on main task accuracy and label recovery error for FedPass.}
\label{tab:gaussian-var}
\end{table}

\section{The General Adaptive Obfuscation in VFL}
Consider a neural network $f_\Theta(x):\calX \to \RR$, where $x \in \calX$, $\Theta$ denotes model parameters of neural networks. In the VFL, $K$ passive parties and one active party collaboratively optimize $\Theta = (\omega, \theta_1, \cdots, \theta_K)$ of network according to Eq. \eqref{eq:loss-VFL-app}.
\begin{equation}\label{eq:loss-VFL-app}
\begin{split}
        \min_{\omega, \theta_1, \cdots, \theta_K} &\frac{1}{n}\sum_{i=1}^n\ell(F_{\omega} \circ (G_{\theta_1}(x_{1,i}),G_{\theta_2}(x_{2,i}), \\
        & \cdots,G_{\theta_K}(x_{K,i})), y_{i}),
\end{split}
\end{equation}

Define the adaptive obfuscation module on both the passive and active party to be: $\{g_{\theta_k, \theta_k'}( s^{p_k},)\}_{k=1}^K,g_{\omega, \omega'}(s^a,)$, where $s^{p_k}$ and $s^a$ are the controlled parameters of the passive party $j$ and the active party respectively, $\{\theta_k'\}_{k=1}^K$ and $\omega'$ are additional learnable parameters $g()$ introduce. The adaptive obfuscation framework in VFL aims to optimize:
\begin{equation} \label{eq:loss-aof}
\begin{split}
        \min_{\omega,\omega', \theta_1, \theta_1',\cdots, \theta_K,\theta_K'} &\frac{1}{N}\sum_{i=1}^N\ell(F_{\omega}g_{\omega, \omega'} \circ (G_{\theta_1}(g_{\theta_1, \theta_1'}(x_{1,i}, s^{p_1})), \\
        &\cdots, G_{\theta_K}(g_{\theta_K,\theta_K'}(x_{K,i},s^{p_K})),s^a, y_{i}),
\end{split}
\end{equation}
Denote the composite function $F_{\omega}\cdot g_{\omega, \omega'}()$ as $F'_{\omega, \omega'}()$ and $G_{\theta_j}g_{\theta_j, \theta_j'}()$ as $G'_{\theta_j, \theta_j'}()$, $j=1, \cdots, K$. Then we rewrite the Eq. \eqref{eq:loss-aof} as follows 
\begin{equation} \label{eq:loss-aof-2}
\begin{split}
        \min_{\omega,\omega', \theta_1, \theta_1',\cdots, \theta_K,\theta_K'} &\frac{1}{N}\sum_{i=1}^N\ell(F'_{\omega, \omega'} \circ (G'_{\theta_1,\theta_1'}(x_{1,i}, s^{p_1})), \\
        &\cdots, G'_{\theta_K,\theta_K'}(x_{K,i}, s^{p_K})) , s^a, y_{i}),
\end{split}
\end{equation}
Minimizing the Eq. \eqref{eq:loss-aof-2} could be regarded as two optimization procedures with the same loss. One is to maximize model performance w.r.t model parameters $\omega, \theta_1, \cdots, \theta_K$ given the training feature $x$ and label $y$. The second is to train the model parameters $\omega,\omega', \theta_1, \theta_1',\cdots, \theta_K,\theta_K'$ with main task loss to let the adaptive obfuscation module  $\{g_{\theta_k, \theta_k'}( s^{p_k},)\}_{k=1}^K,g_{\omega, \omega'}(s^a,)$ fit the main task.
\subsection{Derivative of Adaptive Obfuscation for FedPass}
Specifically, we take the FedPass as one example, the adaptive obfuscation function is the following:
\begin{equation}\label{eq:pst1-app}
\begin{split}
        g_{W,W'}(x_{in}, s) = &
        \gamma(Wx_{in}) + \beta, \\
         \gamma=&\text{Avg}\Big( D\big(E(Ws_\gamma)\big)\Big)\\
        \beta = &\text{Avg}\Big( D\big(E(Ws_\beta)\big)\Big)
\end{split}
\end{equation}
where $W$ denotes the model parameters of the neural network layer for inserting passports, $x_{in}$ is the input fed to $W$, $\gamma$ and $\beta$ are the scale factor and the bias term. Note that the determination of the crucial parameters $\gamma$ and $\beta$ involves the model parameter $W$ with private passports $s_\gamma$ and $s_\beta$, followed by a autoencoder (Encoder $E$ and Decoder $D$ with parameters $W'$). The the derivative for $g$ w.r.t. $W$ has the following three backpropagation paths via $\beta, \gamma, W$:
\begin{equation} \label{eq:upload-gradients-app1}
\frac{\partial g}{\partial W} = \left\{
\begin{aligned}
x_{in} \otimes diag(\gamma)^T + \beta \quad & W \text{ path}\\
(Wx_{in})^T\frac{\partial \gamma}{\partial W}  \quad & \text{$\gamma$ path}\\
\frac{\partial \beta}{\partial W},  \quad & \text{$\beta$ path}\\
\end{aligned}
\right.
\end{equation}
where $\otimes$ represents Kronecker product.
Moreover, the derivative for $g$ w.r.t. $W'$ is:
\begin{equation} \label{eq:upload-gradients-app2}
\frac{\partial g}{\partial W'} = \left\{
\begin{aligned}
\frac{\partial g}{\partial \gamma}\frac{\partial \gamma}{\partial W'}  \quad & \text{$\gamma$ path}\\
\frac{\partial g}{\partial \beta}\frac{\partial \beta}{\partial W'}  \quad & \text{$\beta$ path}\\
\end{aligned}
\right.
\end{equation}




\section{Proof}
We investigate the privacy-preserving capability of FedPass against feature reconstruction attack and label inference attack. Note that we conduct the privacy analysis with linear regression models, for the sake of brevity. 

\begin{definition} \label{def:SplitFed-app}
Define the forward function of the passive model $G$ and the active model $F$: 
\begin{itemize}
    \item For passive layer: $H = G(x) =  W_p s_\gamma^p \cdot W_p x + W_p s_\beta^p$.
    \item For active layer: $y = F(H) =  W_a s_\gamma^a \cdot W_a  H + W_a s_\beta^a$.
\end{itemize}
where $W_p$, $W_a$ are 2D matrices of the passive and active models; $\cdot$ denotes the inner product, $ s_\gamma^p,  s_\beta^p$ are passports of the passive party, $ s_\gamma^a,  s_\beta^a$ are passports of the active party.

\end{definition}

\subsubsection{Hardness of feature Restoration with FedPass}
Consider the white-box Model Inversion (MI) attack (i.e., model inversion step in CAFE \cite{jin2021cafe,he2019model}) that aims to inverse the model $W_p$ to recover features $\hat{x}$ approximating original features $x$. In this case, the attacker (i.e., the active party) knows the passive model parameters $W_p$, forward embedding $H$ and the way of embedding passport, but does not know the passport.

\begin{lem} \label{lem1-app}
Suppose a passive model as the Def. \ref{def:SplitFed-app} illustrates, an attack estimate the private feature by guessing passports $s_{\gamma'}^p$ and $s_{\beta'}^p$. If $W_p$ is invertible, We could obtain the difference between $x$ and estimated $\hat{x}$ by the adversary in the following two cases:
\begin{itemize}
    \item When inserting the $s_\gamma^p$ for the passive party,
    \begin{equation}
    \|x - \hat{x}\|_2 \geq  \frac{\|(D_\gamma^{-1}-D_{\gamma'}^{-1}) H\|_2}{\| W_p\|_2}.
\end{equation}
\item When inserting the $s_\beta^p$ for the passive party,
\begin{equation} \label{eq:beta-lem-app}
     \|x - \hat{x}\|_2 = \|D_\beta-D_{\beta'}\|_2,
\end{equation}
\end{itemize}
where $D_\gamma=diag(W_ps_\gamma^p),D_{\gamma'}=diag(W_ps_{\gamma'}^p),D_\beta=diag(W_ps_\beta^p),D_{\beta'}=diag(W_ps_{\beta'}^p)$ and $W_p^\dag$ is the Moore–Penrose inverse of $W_p$. 
\end{lem}
\begin{proof}
For inserting the $s_\beta^p$,
\begin{equation*}
    H = W_p s_\gamma^p * W_p x.
\end{equation*}
The attacker with knowing $W_p$ has
\begin{equation*}
    H =W_p s_{\gamma'}^p * W_p \hat{x}.
\end{equation*}
Denote $D_\gamma=diag(W_ps_\gamma^p),D_{\gamma'}=diag(W_ps_{\gamma'}^p)$. If $A$ is invertible, due to
\begin{equation*}
    \|A\|_2\|A^{-1}B\|_2 \leq \|AA^{-1}B\|_2 = \|B\|_2,
\end{equation*}
we have 
\begin{equation*}
    \|A^{-1}B\|_2 \geq \frac{\|B\|_2}{\|A\|_2}.
\end{equation*}   
Therefore, if $W_P^\dag$ is the Moore–Penrose inverse of $W_p$ and $W_p$ is invertible, we can obtain
\begin{equation}\label{eq:case1}
    \begin{split}
        \|\hat{x}-x\|&=\|W_p^\dag D_\gamma^{-1}H-W_p^\dag D_{\gamma'}^{-1}H\|_2 \\
        &= \| W_p^\dag(D_\gamma^{-1}-D_{\gamma'}^{-1})H \|_2 \\
        &\geq \frac{\|(D_\gamma^{-1}-D_{\gamma'}^{-1}) H\|_2}{\| W_p\|_2}.
    \end{split}
\end{equation}
On the other hand, for inserting the $s_\beta^p$ of passive party,
\begin{equation*}
     H =  W_p x + W_p s_\beta^p.
\end{equation*}
The attacker with knowing $W_p$ has
\begin{equation*}
    H  = W_p x + W_p s_{\beta'}^p.
\end{equation*}
We further obtain
\begin{equation}\label{eq:case2}
    \begin{split}
        \|\hat{x}-x\|&=\|W_p^\dag (H-D_\beta)-W_p^\dag (H-D_{\beta'})\|_2 \\
        &= \| W_p^\dag(W_p {s_\beta}^p- W_p s_{\beta'}^p) \|_2 \\
        &= \|s_\beta^p - s_{\beta'}^p \|_2.
    \end{split}
\end{equation}
\end{proof}
\begin{rmk}
If the adversary don't know the existence of passport layer, then $D_{\gamma'} = I$ and $s_{\beta'}^p =0$ in Eq. \eqref{eq:case1} and \eqref{eq:case2}.
\end{rmk}
\begin{rmk}
Eq. \eqref{eq:beta-lem-app} only needs $W_p$ is a left inverse, i.e., $W_p$ has linearly independent columns. 
\end{rmk}

Lemma \ref{lem1-app} illustrates that the difference between estimated feature and original feature has the lower bound, which depends on the difference between original passport and inferred passport. Specifically, if the inferred passport is far away from $\|s_\beta^p - s_{\beta'}^p \|_2$ and $\|(D_\gamma^{-1}-D_{\gamma'}^{-1})\|$ goes large causing a large reconstruction error by attackers. Furthermore, we provide the analysis of the probability of attackers to reconstruct the $x$ in Theorem \ref{thm:thm1-app}. Let $m$ denote the dimension of the passport via flattening, $N$ denote the passport range formulated in Eq. \eqref{eq:sample-pst-app} and $\Gamma(\cdot)$ denote the Gamma distribution.

\begin{theorem}\label{thm:thm1-app}
    Suppose the passive party protects features $x$ by inserting the $s_{\beta}^p$. The probability of recovering features by the attacker via white-box MI attack is at most $\frac{\pi^{m/2}\epsilon^m}{\Gamma(1+m/2)N^m}$ such that the recovering error is less than $\epsilon$, i.e., $\|x-\hat{x}\|_2\leq \epsilon$,
\end{theorem}

\begin{proof}
    According to Lemma \ref{lem1-app}, the attacker aims to recover the feature $\hat{x}$ within the $\epsilon$ error from the original feature $x$, that is, the guessed passport needs to satisfy: 
    \begin{equation}
        \|s_{\beta'}^p - s_\beta^p  \|_2 \leq \epsilon
    \end{equation}
Therefore, the area of inferred passport of attackers is sphere with the center $s_{\beta}^p$ and radius $\epsilon$. And the volumes of this area is at most $\frac{\pi^{m/2}\epsilon^m}{\Gamma(1+m/2)}$, where $F$ represent the Gamma distribution and $m$ is dimension of the passport. Consequently, the probability of attackers to successfully recover the feature within $\epsilon$ error is:
\begin{equation*}
\begin{split}
        p \leq \frac{\pi^{m/2}\epsilon^m}{\Gamma(1+m/2)N^m},
        \end{split}
\end{equation*}
where the whole space is $(-N,0)^m$ with volume $N^m$.
\end{proof}
Theorem \ref{thm:thm1-app} demonstrates that the attacker's probability of recovering features within error $\epsilon$ is exponentially small in the dimension of passport size $m$. The successful recovering probability is inversely proportional to the passport range $N$, which is consistent with our ablation study in Sect. \ref{sec:aba-tradeoff}.
\begin{rmk}
The attacker's behaviour we consider in Theorem \ref{thm:thm1-app} is that they guess the private passport randomly. 
\end{rmk}
\subsubsection{Hardness of Label Recovery with FedPass}
Consider the passive model competition attack \cite{fu2022label} that aims to recover labels owned by the active party. The attacker (i.e., the passive party) leverages a small auxiliary labeled dataset $\{x_i, y_i\}_{i=1}^{n_a}$ belonging to the original training data to train the attack model $W_{att}$, and then infer labels for the test data. Note that the attacker knows the trained passive model $G$ and forward embedding $H_i = G(x_i)$. Therefore, they optimizes the attack model $W_{att}$ by minimizing $\ell = \sum_{i=1}^{n_a}\|W_{att}H_i-y_i\|_2$.

\begin{assumption}\label{assum1-app}
Suppose the original main algorithm of VFL is convergent. For the attack model, we assume the error of the optimized attack model $W^*_{att}$ on test data $\tilde{\ell}_t$ is larger than that of the auxiliary labeled dataset $\Tilde{\ell}_a$.
\end{assumption}

\begin{rmk}
The test error is usually higher than the training error because the error is computed on an unknown dataset that the model hasn't seen.
\end{rmk}
\begin{theorem} \label{thm2-app}
Suppose the active party protect $y$ by embedding $s^a_\gamma$, and adversaries aim to recover labels on the test data with the error $\Tilde{\ell}_t$ satisfying:
    \begin{equation}
        \Tilde{\ell}_t \geq \min_{W_{att}} \sum_{i=1}^{n_a}\|(W_{att}- T_i)H_i\|_2,
    \end{equation}
    where $T_i =diag(W_as_{\gamma,i}^a) W_a$ and $s_{\gamma,i}^a$ is the passport for the label $y_i$ embedded in the active model. Moreover, if $H_{i_1} = H_{i_2} = H$ for any $1\leq i_1,i_2 \leq n_a$, then
    \begin{equation} \label{eq:protecty-app}
        \Tilde{\ell}_t \geq \frac{1}{(n_a-1)}\sum_{1\leq i_1<i_2\leq n_a}\|(T_{i_1}-T_{i_2})H\|_2
    \end{equation}
\end{theorem}

\begin{proof}
For only inserting the passport $s^a_\gamma$ of active party, according to Assumption \ref{assum1-app}, we have
\begin{equation}
    y_i=   W_a s_\gamma^a \cdot W_a  H_i
\end{equation}
Moreover, the attackers aims to optimize
\begin{equation*}
\begin{split}
       &\min_W\sum_{i=1}^{n_a}\|WH_i-y_i\|_2  \\
       = &\min_W\sum_{i=1}^{n_a}\|WH_i - W_a s_\gamma^a \cdot W_a  H_i\|_2 \\
       =&\min_W\sum_{i=1}^{n_a}\|(W-T_i)H_i\|_2,
\end{split}   
\end{equation*}
where $T_i =diag(W_as_{\gamma,i}^a) W_a$. Therefore, based on Assumption \ref{assum1-app}, $\Tilde{\ell}_t\geq \min_W\sum_{i=1}^{n_a}\|(W-T_i)H_i\|_2$. Moreover, if $H_i = H_j = H$ for any $i_1,i_2 \in [n_a]$, then
\begin{equation*}
    \begin{split}
        \Tilde{\ell}_t& \geq \min_W\sum_{i=1}^{n_a}\|(W-T_i)H_i\|_2 \\
        &= \min_W \frac{1}{2(n_a-1)}\sum_{1\leq i_1,i_2\leq n_a}(\|(W-T_{i_1})H\|_2 \\
        &  \qquad  + \|(W-T_{i_2})H\|_2) \\
        &\geq \min_W \frac{1}{2(n_a-1)}\sum_{1\leq i_1,i_2\leq n_a}(\|(T_{i_1}-T_{i_2})H\|_2) \\
       & =  \frac{1}{(n_a-1)}\sum_{1\leq i_1<i_2\leq n_a}(\|(T_{i_1}-T_{i_2})H\|_2)
    \end{split}
\end{equation*}
\end{proof}

\begin{prop}\label{prop1-app}
Since passports are randomly generated and $W_a$ and $H$ are fixed, if the $W_a = I, H=\Vec{1}$, then it follows that:
\begin{equation}
    \Tilde{\ell}_t \geq \frac{1}{(n_a-1)}\sum_{1\leq i_1<i_2\leq n_a}\|s_{\gamma,i_1}^a-s_{\gamma,i_2}^a\|_2)
\end{equation}
\end{prop}
\begin{proof}
  When $W_a =I$ and $H = \Vec{1}$, $(T_{i_1}-T_{i_2})H = s_{\gamma,i_1}^a-s_{\gamma,i_2}$ . Therefore, based on Theorem \ref{thm2-app}, we obtain
  \begin{equation}
          \Tilde{\ell}_t \geq \frac{1}{(n_a-1)}\sum_{1\leq i_1<i_2\leq n_a}\|s_{\gamma,i_1}^a-s_{\gamma,i_2}^a\|_2)
  \end{equation}
\end{proof}

Theorem \ref{thm2-app} and Proposition \ref{prop1-app} show that the label recovery error $\Tilde{\ell}_t$ has a lower bound,
which deserves further explanations. 
\begin{itemize}
    \item First, when passports are randomly generated for all data, i.e., $s_{\gamma,i_1}^a \neq s_{\gamma,i_2}^a$, then a non-zero label recovery error is guaranteed no matter how adversaries attempt to minimize it. The recovery error thus acts as a protective random noise imposed on true labels.
    \item Second, the magnitude of the recovery error monotonically increases with the variance $\sigma^2$ of the Gaussian distribution passports sample from (It is because the difference of two samples from the same Gaussian distribution $\calN(\mu, \sigma^2)$ in Eq. \eqref{eq:sample-pst-app} depends on the variance $\sigma^2$), which is a crucial parameter to control privacy-preserving capability. Experiments on Sect. \ref{sec:aba-tradeoff} also verify this phenomenon. 
    \item Third, it is worth noting that the lower bound is based on the training error of the auxiliary data used by adversaries to launch PMC attacks. Given possible discrepancies between the auxiliary data and private labels, e.g., in terms of distributions and the number of dataset samples, the actual recovery error of private labels can be much larger than the lower bound. 
\end{itemize}
 

\newpage
\bibliographystyle{named}
\bibliography{related}

\begin{thebibliography}{}

\bibitem[\protect\citeauthoryear{Abadi \bgroup \em et al.\egroup
  }{2016}]{abadi2016deep}
Martin Abadi, Andy Chu, Ian Goodfellow, H~Brendan McMahan, Ilya Mironov, Kunal
  Talwar, and Li~Zhang.
\newblock Deep learning with differential privacy.
\newblock In {\em Proceedings of the 2016 ACM SIGSAC conference on computer and
  communications security}, pages 308--318, 2016.

\bibitem[\protect\citeauthoryear{Aji and Heafield}{2017}]{aji2017sparse}
Alham~Fikri Aji and Kenneth Heafield.
\newblock Sparse communication for distributed gradient descent.
\newblock In {\em Proceedings of the 2017 Conference on Empirical Methods in
  Natural Language Processing}, pages 440--445, 2017.

\bibitem[\protect\citeauthoryear{Cheng \bgroup \em et al.\egroup
  }{2021}]{secureboost}
K.~Cheng, T.~Fan, Y.~Jin, Y.~Liu, T.~Chen, D.~Papadopoulos, and Q.~Yang.
\newblock Secureboost: A lossless federated learning framework.
\newblock {\em IEEE Intelligent Systems}, 36(06):87--98, nov 2021.

\bibitem[\protect\citeauthoryear{Dryden \bgroup \em et al.\egroup
  }{2016}]{dryden2016communication}
Nikoli Dryden, Tim Moon, Sam~Ade Jacobs, and Brian Van~Essen.
\newblock Communication quantization for data-parallel training of deep neural
  networks.
\newblock In {\em 2016 2nd Workshop on Machine Learning in HPC Environments
  (MLHPC)}, pages 1--8. IEEE, 2016.

\bibitem[\protect\citeauthoryear{Fan \bgroup \em et al.\egroup
  }{2019}]{fan2019rethinking}
Lixin Fan, Kam~Woh Ng, and Chee~Seng Chan.
\newblock Rethinking deep neural network ownership verification: Embedding
  passports to defeat ambiguity attacks.
\newblock {\em Advances in neural information processing systems}, 32, 2019.

\bibitem[\protect\citeauthoryear{Fan \bgroup \em et al.\egroup
  }{2020}]{fan2020rethinking}
Lixin Fan, Kam~Woh Ng, Ce~Ju, Tianyu Zhang, Chang Liu, Chee~Seng Chan, and
  Qiang Yang.
\newblock Rethinking privacy preserving deep learning: How to evaluate and
  thwart privacy attacks.
\newblock In {\em Federated Learning}, pages 32--50. Springer, 2020.

\bibitem[\protect\citeauthoryear{Fan \bgroup \em et al.\egroup
  }{2021}]{fan2021deepip}
Lixin Fan, Kam~Woh Ng, Chee~Seng Chan, and Qiang Yang.
\newblock Deepip: Deep neural network intellectual property protection with
  passports.
\newblock {\em IEEE Transactions on Pattern Analysis \& Machine Intelligence},
  (01):1--1, 2021.

\bibitem[\protect\citeauthoryear{Fu \bgroup \em et al.\egroup
  }{2022a}]{fu2022label}
Chong Fu, Xuhong Zhang, Shouling Ji, Jinyin Chen, Jingzheng Wu, Shanqing Guo,
  Jun Zhou, Alex~X Liu, and Ting Wang.
\newblock Label inference attacks against vertical federated learning.
\newblock In {\em 31st USENIX Security Symposium (USENIX Security 22), Boston,
  MA}, 2022.

\bibitem[\protect\citeauthoryear{Fu \bgroup \em et al.\egroup
  }{2022b}]{blindFL}
Fangcheng Fu, Huanran Xue, Yong Cheng, Yangyu Tao, and Bin Cui.
\newblock Blindfl: Vertical federated machine learning without peeking into
  your data.
\newblock In {\em Proceedings of the 2022 International Conference on
  Management of Data}, SIGMOD '22, page 1316–1330, New York, NY, USA, 2022.
  Association for Computing Machinery.

\bibitem[\protect\citeauthoryear{Gasc{\'o}n \bgroup \em et al.\egroup
  }{2016}]{gascon2016secure}
Adri{\`a} Gasc{\'o}n, Phillipp Schoppmann, Borja Balle, Mariana Raykova, Jack
  Doerner, Samee Zahur, and David Evans.
\newblock Secure linear regression on vertically partitioned datasets.
\newblock {\em IACR Cryptol. ePrint Arch.}, 2016:892, 2016.

\bibitem[\protect\citeauthoryear{Hardy \bgroup \em et al.\egroup
  }{2017}]{hardy2017private}
Stephen Hardy, Wilko Henecka, Hamish Ivey-Law, Richard Nock, Giorgio Patrini,
  Guillaume Smith, and Brian Thorne.
\newblock Private federated learning on vertically partitioned data via entity
  resolution and additively homomorphic encryption.
\newblock {\em arXiv preprint arXiv:1711.10677}, 2017.

\bibitem[\protect\citeauthoryear{He \bgroup \em et al.\egroup
  }{2016}]{he2016deep}
Kaiming He, Xiangyu Zhang, Shaoqing Ren, and Jian Sun.
\newblock Deep residual learning for image recognition.
\newblock In {\em Proceedings of the IEEE conference on computer vision and
  pattern recognition}, pages 770--778, 2016.

\bibitem[\protect\citeauthoryear{He \bgroup \em et al.\egroup
  }{2019}]{he2019model}
Zecheng He, Tianwei Zhang, and Ruby~B Lee.
\newblock Model inversion attacks against collaborative inference.
\newblock In {\em Proceedings of the 35th Annual Computer Security Applications
  Conference}, pages 148--162, 2019.

\bibitem[\protect\citeauthoryear{He \bgroup \em et al.\egroup
  }{2022}]{he2022hybrid}
Yuanqin He, Yan Kang, Jiahuan Luo, Lixin Fan, and Qiang Yang.
\newblock A hybrid self-supervised learning framework for vertical federated
  learning.
\newblock {\em arXiv preprint arXiv:2208.08934}, 2022.

\bibitem[\protect\citeauthoryear{Hu \bgroup \em et al.\egroup
  }{2022}]{hu2022vertical}
Yuzheng Hu, Tianle Cai, Jinyong Shan, Shange Tang, Chaochao Cai, Ethan Song,
  Bo~Li, and Dawn Song.
\newblock Is vertical logistic regression privacy-preserving? a comprehensive
  privacy analysis and beyond.
\newblock {\em arXiv preprint arXiv:2207.09087}, 2022.

\bibitem[\protect\citeauthoryear{Huang \bgroup \em et al.\egroup
  }{2020}]{huang2020instahide}
Yangsibo Huang, Zhao Song, Kai Li, and Sanjeev Arora.
\newblock Instahide: Instance-hiding schemes for private distributed learning.
\newblock In {\em International conference on machine learning}, pages
  4507--4518. PMLR, 2020.

\bibitem[\protect\citeauthoryear{Jin \bgroup \em et al.\egroup
  }{2021}]{jin2021cafe}
Xiao Jin, Pin-Yu Chen, Chia-Yi Hsu, Chia-Mu Yu, and Tianyi Chen.
\newblock Cafe: Catastrophic data leakage in vertical federated learning.
\newblock {\em NeurIPS}, 34:994--1006, 2021.

\bibitem[\protect\citeauthoryear{Kang \bgroup \em et al.\egroup
  }{2022a}]{kang2022prada}
Y.~Kang, Y.~He, J.~Luo, T.~Fan, Y.~Liu, and Q.~Yang.
\newblock Privacy-preserving federated adversarial domain adaptation over
  feature groups for interpretability.
\newblock {\em IEEE Transactions on Big Data}, (01):1--12, jul 2022.

\bibitem[\protect\citeauthoryear{Kang \bgroup \em et al.\egroup
  }{2022b}]{kangyan2022fedcvt}
Yan Kang, Yang Liu, and Xinle Liang.
\newblock {{FedCVT}}: {{Semi-supervised Vertical Federated Learning}} with
  {{Cross-view Training}}.
\newblock {\em ACM Transactions on Intelligent Systems and Technology (TIST)},
  May 2022.

\bibitem[\protect\citeauthoryear{Kang \bgroup \em et al.\egroup
  }{2022c}]{kang2022framework}
Yan Kang, Jiahuan Luo, Yuanqin He, Xiaojin Zhang, Lixin Fan, and Qiang Yang.
\newblock A framework for evaluating privacy-utility trade-off in vertical
  federated learning.
\newblock {\em arXiv preprint arXiv:2209.03885}, 2022.

\bibitem[\protect\citeauthoryear{Krizhevsky \bgroup \em et al.\egroup
  }{2012}]{NIPS2012_c399862d}
Alex Krizhevsky, Ilya Sutskever, and Geoffrey~E Hinton.
\newblock Imagenet classification with deep convolutional neural networks.
\newblock In F.~Pereira, C.J. Burges, L.~Bottou, and K.Q. Weinberger, editors,
  {\em Advances in Neural Information Processing Systems}, volume~25. Curran
  Associates, Inc., 2012.

\bibitem[\protect\citeauthoryear{Krizhevsky \bgroup \em et al.\egroup
  }{2014}]{krizhevsky2014cifar}
Alex Krizhevsky, Vinod Nair, and Geoffrey Hinton.
\newblock The cifar-10 dataset.
\newblock {\em online: http://www. cs. toronto. edu/kriz/cifar. html}, 55(5),
  2014.

\bibitem[\protect\citeauthoryear{LeCun \bgroup \em et al.\egroup
  }{1998}]{lecun1998gradient}
Yann LeCun, L{\'e}on Bottou, Yoshua Bengio, and Patrick Haffner.
\newblock Gradient-based learning applied to document recognition.
\newblock {\em Proceedings of the IEEE}, 86(11):2278--2324, 1998.

\bibitem[\protect\citeauthoryear{LeCun \bgroup \em et al.\egroup
  }{2010}]{lecun2010mnist}
Yann LeCun, Corinna Cortes, and CJ~Burges.
\newblock Mnist handwritten digit database.
\newblock {\em ATT Labs [Online]. Available: http://yann.lecun.com/exdb/mnist},
  2, 2010.

\bibitem[\protect\citeauthoryear{Li \bgroup \em et al.\egroup
  }{2020}]{li2020review}
Li~Li, Yuxi Fan, Mike Tse, and Kuo-Yi Lin.
\newblock A review of applications in federated learning.
\newblock {\em Computers \& Industrial Engineering}, 149:106854, 2020.

\bibitem[\protect\citeauthoryear{Li \bgroup \em et al.\egroup
  }{2022a}]{li2022fedipr}
Bowen Li, Lixin Fan, Hanlin Gu, Jie Li, and Qiang Yang.
\newblock Fedipr: Ownership verification for federated deep neural network
  models.
\newblock {\em IEEE Transactions on Pattern Analysis and Machine Intelligence},
  2022.

\bibitem[\protect\citeauthoryear{Li \bgroup \em et al.\egroup
  }{2022b}]{oscar2022split}
Oscar Li, Jiankai Sun, Xin Yang, Weihao Gao, Hongyi Zhang, Junyuan Xie,
  Virginia Smith, and Chong Wang.
\newblock Label leakage and protection in two-party split learning.
\newblock In {\em International Conference on Learning Representations}, 2022.

\bibitem[\protect\citeauthoryear{Lin \bgroup \em et al.\egroup
  }{2018}]{lin2018deep}
Yujun Lin, Song Han, Huizi Mao, Yu~Wang, and Bill Dally.
\newblock Deep gradient compression: Reducing the communication bandwidth for
  distributed training.
\newblock In {\em International Conference on Learning Representations}, 2018.

\bibitem[\protect\citeauthoryear{Liu \bgroup \em et al.\egroup
  }{2020}]{liu2020secure}
Yang Liu, Yan Kang, Chaoping Xing, Tianjian Chen, and Qiang Yang.
\newblock A secure federated transfer learning framework.
\newblock {\em IEEE Intelligent Systems}, 35(4):70--82, 2020.

\bibitem[\protect\citeauthoryear{Liu \bgroup \em et al.\egroup
  }{2021}]{liu2021defending}
Yang Liu, Zhihao Yi, Yan Kang, Yuanqin He, Wenhan Liu, Tianyuan Zou, and Qiang
  Yang.
\newblock Defending label inference and backdoor attacks in vertical federated
  learning.
\newblock {\em arXiv preprint arXiv:2112.05409}, 2021.

\bibitem[\protect\citeauthoryear{Liu \bgroup \em et al.\egroup
  }{2022a}]{liu2022vertical}
Yang Liu, Yan Kang, Tianyuan Zou, Yanhong Pu, Yuanqin He, Xiaozhou Ye,
  Ye~Ouyang, Ya-Qin Zhang, and Qiang Yang.
\newblock Vertical federated learning.
\newblock {\em arXiv preprint arXiv:2211.12814}, 2022.

\bibitem[\protect\citeauthoryear{Liu \bgroup \em et al.\egroup
  }{2022b}]{liu2022cross}
Yang Liu, Xinle Liang, Jiahuan Luo, Yuanqin He, Tianjian Chen, Quanming Yao,
  and Qiang Yang.
\newblock Cross-silo federated neural architecture search for heterogeneous and
  cooperative systems.
\newblock In {\em Federated and Transfer Learning}, pages 57--86. Springer,
  2022.

\bibitem[\protect\citeauthoryear{Luo \bgroup \em et al.\egroup
  }{2021}]{Luo2021fi}
Xinjian Luo, Yuncheng Wu, Xiaokui Xiao, and Beng~Chin Ooi.
\newblock Feature inference attack on model predictions in vertical federated
  learning.
\newblock In {\em 2021 {IEEE} 37th International Conference on Data Engineering
  ({ICDE})}. {IEEE}, apr 2021.

\bibitem[\protect\citeauthoryear{Wu \bgroup \em et al.\egroup
  }{2015}]{wu20153d}
Zhirong Wu, Shuran Song, Aditya Khosla, Fisher Yu, Linguang Zhang, Xiaoou Tang,
  and Jianxiong Xiao.
\newblock 3d shapenets: A deep representation for volumetric shapes.
\newblock In {\em Proceedings of the IEEE conference on computer vision and
  pattern recognition}, pages 1912--1920, 2015.

\bibitem[\protect\citeauthoryear{Yang \bgroup \em et al.\egroup
  }{2019}]{yang2019federated}
Qiang Yang, Yang Liu, Tianjian Chen, and Yongxin Tong.
\newblock Federated machine learning: Concept and applications.
\newblock {\em ACM Transactions on Intelligent Systems and Technology (TIST)},
  10(2):1--19, 2019.

\bibitem[\protect\citeauthoryear{Zhang \bgroup \em et al.\egroup
  }{2018}]{zhang2018mixup}
Hongyi Zhang, Moustapha Cisse, Yann~N Dauphin, and David Lopez-Paz.
\newblock mixup: Beyond empirical risk minimization.
\newblock In {\em International Conference on Learning Representations}, 2018.

\bibitem[\protect\citeauthoryear{Zhang \bgroup \em et al.\egroup
  }{2022}]{zhang2022trading}
Xiaojin Zhang, Yan Kang, Kai Chen, Lixin Fan, and Qiang Yang.
\newblock Trading off privacy, utility and efficiency in federated learning.
\newblock {\em arXiv preprint arXiv:2209.00230}, 2022.

\bibitem[\protect\citeauthoryear{Zhu \bgroup \em et al.\egroup
  }{2019}]{zhu2019dlg}
Ligeng Zhu, Zhijian Liu, , and Song Han.
\newblock Deep leakage from gradients.
\newblock In {\em Annual Conference on Neural Information Processing Systems
  (NeurIPS)}, 2019.

\bibitem[\protect\citeauthoryear{Zou \bgroup \em et al.\egroup
  }{2022}]{zou2022defending}
Tianyuan Zou, Yang Liu, Yan Kang, Wenhan Liu, Yuanqin He, Zhihao Yi, Qiang
  Yang, and Ya-Qin Zhang.
\newblock Defending batch-level label inference and replacement attacks in
  vertical federated learning.
\newblock {\em IEEE Transactions on Big Data}, 2022.

\end{thebibliography}

\end{document}